\RequirePackage[T1]{fontenc}
\RequirePackage{fix-cm} % See https://ftp.jaist.ac.jp/pub/CTAN/macros/latex/base/fix-cm.pdf
\RequirePackage{mlmodern} % for a thicker version of the TeX standard font
\documentclass[11pt,reqno,letterpaper]{amsart}
\usepackage[latin1]{inputenc}
\usepackage[english,activeacute]{babel}
\usepackage{mathtools}
\usepackage{amsmath,amssymb,amsthm}
\usepackage[colorlinks=true,citecolor=black,linkcolor=black,urlcolor=blue]{hyperref}
\usepackage{url}
\usepackage{pgfplots}
\usepackage{xcolor}
\usepackage{subcaption}
\pgfplotsset{compat=1.17}
\usepackage{cleveref}

\newtheorem{thm}{Theorem}[section]
\newtheorem{cor}[thm]{Corollary}
\newtheorem{pro}[thm]{Proposition}
\newtheorem{lem}[thm]{Lemma}

\theoremstyle{definition}
	\newtheorem{defi}[thm]{Definition}
	\newtheorem{exa}[thm]{Example}
	
	\newtheorem{rem}[thm]{Remark}

\numberwithin{equation}{section}

\DeclareMathOperator{\dis}{d}

\DeclareMathOperator{\ev}{ev}

\DeclareMathOperator{\tr}{tr}
\DeclareMathOperator{\trs}{trs}

\DeclareMathOperator{\card}{card}
\DeclareMathOperator{\spn}{span}
\DeclareMathOperator{\BCH}{BCH}

\newcommand{\F}{\mathbb{F}}

\newcommand{\Z}{\mathbb{Z}}

\DeclarePairedDelimiter\bra{\langle}{\rvert}
\DeclarePairedDelimiter\ket{\lvert}{\rangle}

\title[Entanglement assisted quantum $(r,\delta)$-locally recoverable codes]{Entanglement assisted quantum $(r,\delta)$-locally recoverable codes}

\author{C. Galindo, F. Hernando, H. Mart\'in-Cruz and R. Matsumoto}
\curraddr{\texttt{Carlos Galindo and Fernando Hernando:} Instituto
	Universitario de Matem\'aticas y Aplicaciones de Castell\'on and
	Departamento de Matem\'aticas, Universitat Jaume I, Campus de Riu
	Sec. 12071 Castell\'{o}, Spain\\
	\texttt{Helena Mart\'{\i}n-Cruz:} Departamento de Matem\'aticas, Universidad de Ja\'en, Campus Las Lagunillas 23071 Ja\'en, Spain\\
	\texttt{Ryutaroh Matsumoto:} Department of Information and Communications Engineering, Institute of Science Tokyo, Japan.}

\email{{\rm Galindo: galindo@uji.es; {\rm Hernando:} carrillf@uji.es; {\rm Mart\'{\i}n-Cruz:} hmartin@ujaen.es; {\rm Matsumoto:} ryutaroh@ict.e.titech.ac.jp}}

\urladdr{{\rm Galindo: 0000-0002-3908-4462; {\rm Hernando:} 0000-0002-9758-2152; {\rm Mart\'{\i}n-Cruz:} 0000-0002-6379-6902; {\rm Matsumoto:} 0000-0002-5085-8879}\color{black}}

\thanks{{\it Acknowledgements:} The first three authors were partially funded by MICIU/AEI/10.13039/501100011033 and by ``ERDF, UE'' (grant PID2022-138906NB-C22). The fourth author was partially supported by ``Japan Society for the Promotion of Science'' (grant 23K10980).}

\begin{document}

\begin{abstract}
Quantum $(r,\delta)$-locally recoverable codes are quantum error-correcting codes capable of correcting $\delta-1$ qudit erasures within one subset of qudits of cardinality at most $r+\delta-1$. In this paper, we introduce the more general framework of entanglement-assisted quantum $(r,\delta)$-locally recoverable codes, assuming that the local recovery operation is assisted by receiver-held qudits that remain unaffected by erasures.

We establish necessary and sufficient conditions for these codes to satisfy this property. For codes derived from Hermitian or Euclidean constructions, we establish connections between entanglement-assisted quantum and classical notions of $(r,\delta)$-local recoverability, and derive a Singleton-like bound. Furthermore, we construct optimal pure entan\-gle\-ment-assisted quantum $(r,\delta)$-locally recoverable codes from several families of classical codes, including bivariate $J$-affine variety codes, BCH codes, and homothetic-BCH codes.

\end{abstract}

\maketitle

\section{Introduction}

Recent evidence of quantum supremacy \cite{Aru, BallP, Zhong, LiuY} highlights the potential of quantum information processing across numerous practical applications. Protecting quantum states from noise and decoherence is essential to ensure the reliability of quantum devices. Although the no-cloning theorem dictates that quantum information cannot be cloned \cite{Dieks, Woot}, quantum error correction remains viable \cite{ShorS, 95kkk}, and \textit{quantum error-correcting codes} (QECCs) are designed for this purpose.

The initial work on QECCs was restricted to the binary case \cite{20kkk, Gottesman, Calder2, Calderbank}, but the theory was later extended to non-binary systems. Since then, a wide variety of algebraic methods have been developed to construct non-binary quantum codes (see, for instance, \cite{BE, AK, Ketkar, Aly, XingC, Lag2, gahe, Anderson}). Given a prime power $q$, a $q$-ary QECC of length $n$ is a $K$-dimensional subspace of the Hilbert space $\mathbb{C}^{q^n} \cong (\mathbb{C}^q)^{\otimes n}$. The parameters of such a QECC are denoted by $((n,K,d))_q$, where the minimum distance $d$ represents the minimum weight of an undetectable error operator acting on the system.

\textit{(Quantum) stabilizer codes} constitute an important class of QECCs because they are associated with classical additive codes. Specifically, an $((n,K,d))_q$ stabilizer code exists if and only if there is an additive classical code $C \subseteq \F_q^{2n}$ that is self-orthogonal under the trace-symplectic form, where $\F_q$ denotes the $q$-ary finite field (see \Cref{thm:equiv} for more details). Algebraic constructions can be derived using Hermitian or Euclidean inner products over $\F_{q^2}^n$ or $\F_q^n$, respectively.

The restrictive self-orthogonality condition mentioned above (or dual-containment when considering the corresponding dual code) prevents many standard classical codes from being converted into quantum stabilizer codes. To overcome this limitation, Brun et al. \cite{Brun} introduced the framework of entanglement-assisted quantum error correction (EAQEC) by pre-sharing maximally entangled states between the sender and receiver. This approach simplifies quantum error-correcting theory and enhances transmission capacity. Moreover, it enables the construction of \textit{entanglement-assisted QECCs (EAQECCs)} from arbitrary linear codes via symplectic forms, or Hermitian or Euclidean inner products. The parameters of an EAQECC are denoted $[[n,k,d;c]]_q$, where $c$ represents the number of required maximally entangled pairs. Formulas for these parameters in the non-binary setting are given in \cite{galindo19} and recalled in \Cref{thm:csh}.

%\medskip

Many major technology companies store data in large-scale cloud storage and distributed computing systems, where server and node failures are routine challenges. To efficiently restore lost data without accessing the entire system, classical \textit{$r$-locally recoverable codes} (\textit{$r$-LRCs} or \textit{LRCs with locality $r$}) were introduced by Gopalan et al. \cite{GHSY2012}. These codes enable the recovery of a single erasure by accessing at most $r$ other symbol positions. Recent studies on $r$-LRCs include \cite{TPD2016, BTV2017, LMC2018, Mi2018, LXY2019, J2019, LMT2020, SVV2021, edgar}. Subsequently, Prakash et al. \cite{PKLK2012} introduced \textit{$(r,\delta)$-locally recoverable codes} (\textit{$(r,\delta)$-LRCs} or \textit{LRCs with locality $(r,\delta)$}) to handle more realistic scenarios involving multiple simultaneous node failures. These codes can repair up to $\delta - 1$ erasures by accessing at most $r + \delta - 1$ positions, including the erased ones. Recent works in this direction include \cite{SDYL2014, CXHF2018, LMX2019, LXY2019, SZW2019, J2019, CFXF2019, Z2020, FF2020, QZF2021, KWG2021, Cai, GFMC}. Both $r$-LRCs and $(r,\delta)$-LRCs obey Singleton-like bounds on their parameters and locality (see \Cref{Singleform} for the $(r,\delta)$-LRC case). Codes attaining these bounds are termed \textit{optimal}, and many of the aforementioned references construct families of optimal LRCs.

Driven by rapid advancements in quantum computing and the growing scale of quantum information processing, \textit{quantum LRCs} have emerged as the natural quantum counterpart to classical LRCs. Furthermore, substantial industry investment in quantum technology makes quantum data storage an increasingly viable future paradigm. A foundational step in this direction was taken by Golowich and Guruswami \cite{Golowich,Golowich2}, who introduced \textit{quantum $r$-locally recoverable codes} (\textit{quantum $r$-LRCs}) to address single-qudit erasures. Subsequent developments on these codes include \cite{Sharma, luo25, li25, li26}.

The quantum analogue of classical local recovery under multiple erasures was introduced in \cite{qlrc24} and is recalled in \Cref{sub:cqlrcs}. Conceptually, a QECC $Q$ of length $n$ is $(I,K)$-locally recoverable if the action of the completely depolarizing channel on the qudits in $I \subsetneq K$ can be reversed by inspecting the remaining qudits in $K$. A code $Q$ is then defined as a quantum $(r,\delta)$-LRC if every index $i \in \{1,\ldots,n\}$ belongs to a set $K$ of cardinality at most $r + \delta - 1$ such that $Q$ is $(I,K)$-locally recoverable for every subset $I \subset K$ of cardinality $\delta - 1$. Furthermore, the authors of \cite{qlrc24} provide a necessary and sufficient condition for a quantum stabilizer code to be $(r,\delta)$-locally recoverable. This condition depends solely on the puncturing and shortening at suitable coordinate sets of both the underlying symplectic self-orthogonal code and its symplectic dual. Moreover, when using Hermitian or Euclidean inner products, they establish an equivalence between the classical and quantum notions of $(r,\delta)$-local recoverability under a mild additional condition.

Various parameter bounds for quantum LRCs have been established. The first bound for quantum $r$-LRCs was derived by Golowich and Guruswami \cite[Theorem 36]{Golowich} (note that this bound does not appear in the conference version), with subsequent refinements provided in \cite{luo25,li25,li26}, though some of these apply exclusively to pure codes. A Singleton-like bound for pure quantum $(r,\delta)$-LRCs derived from Hermitian or Euclidean inner products was introduced in \cite{qlrc24}. More recently, \cite{impureqlrcs} presented a family of impure quantum codes that exceed the performance bounds for pure codes in quantum local recovery.

%\medskip

A formal framework unifying entanglement assistance with $(r,\delta)$-local recoverability remained unexplored. The special case of recovery with locality $r$ under a single erasure ($\delta=2$) was introduced only very recently, and was restricted to CSS codes \cite{EAQLR-r}. In this paper, we bridge this gap for $(r,\delta)$-local recoverability by introducing \textit{entanglement-assisted quantum
$(r,\delta)$-locally recoverable codes} (\textit{EA $(r,\delta)$-QLRCs});
see \Cref{def:eaqlrc}. This framework relies on the notion of \textit{EA
quantum $(I,K)$-local recoverability} (\Cref{def:eaIKlrc}) and generalizes
the definition of standard quantum $(r,\delta)$-LRCs given in
\cite{qlrc24}. Notably, the quantum operation for local recovery is
assisted by $c$ qudits held by the receiver, which are assumed to be
unaffected by erasures (see \Cref{def:eaIKlrc} and \Cref{rem:assistedrdelta}).

Concurrently with the completion of this article, a general framework for $(r,\delta)$-local recoverability with availability $t$ was introduced in \cite{EAQLR-gret}. The definition of $(r,\delta)$-local recoverability with availability $t$ includes our forthcoming Definition \ref{def:eaqlrc} as a special case, as it applies to any quantum code possessing a unitary encoder. Furthermore, \cite{EAQLR-gret} proposes sufficient conditions in terms of the underlying linear codes for the $(r,\delta)$-locality with availability $t$ of a quantum stabilizer code constructed via the Euclidean inner product using two (possibly different) linear codes. These sufficient conditions can be immediately deduced from our forthcoming Theorem \ref{thm:1}. The authors of \cite{EAQLR-gret} proposed several random and non-random constructions of EA-QLRCs; every non-random construction uses a single linear code, satisfies assumption (2) of Corollary \ref{coro3}, and is consequently bounded by our proposed Singleton-like bound when the availability $t$ is set to $1$ (see \Cref{rem:hopelessAGcodes}). Since all of the codes constructed in Sections \ref{sec:Jaf} and \ref{sec:BCHhomo} attain our Singleton-like bound with equality, our codes are at least as good as theirs for single availability. Our code constructions are very different from those in \cite{EAQLR-gret} and, additionally, allow the use of the Hermitian inner product.

An EAQECC $Q'(C)$ is constructed from classical linear codes $C \subseteq \mathbb{F}_q^{2n}$ and $C' \subseteq \mathbb{F}_q^{2(n+c)}$ such that $C' \subseteq C'^{\perp_s}$, where $\perp_s$ denotes the symplectic dual, and $C$ is the projection of $C'$ onto the coordinates indexed by $\{1,\dots,n\}$. By applying the established results for standard quantum $(I,K)$-LRCs to $C'$, we prove in \Cref{thm:1} a necessary and sufficient condition for the EAQECC $Q'(C)$ to be a quantum $(I,K)$-LRC. This condition is formulated in terms of certain punctured and shortened codes derived from $C$ and $C^{\perp_s}$. Note that $Q'(C)$ represents the quantum code associated with $C'$, although $C'$ itself is omitted from the notation.

\Cref{thm:1} can be translated into its counterparts in \Cref{cor:1} using Hermitian or Euclidean inner products. Moreover, by assuming dual-containment instead of self-orthogonality for the code $C'$, that is, replacing $C$ in the preceding discussion with its Hermitian or Euclidean dual, one can establish a bridge between classical and entanglement-assisted quantum $(r,\delta)$-local recoverability (see Corollaries \ref{cor:1} and \ref{coro2}). Within this setting, we also extend a Singleton-like bound to the entanglement-assisted context (see \Cref{coro3}). Entan\-gle\-ment-assisted pure quantum $(r,\delta)$-locally recoverable codes attaining this bound are called optimal.

\Cref{sec:Jaf,sec:BCHhomo} provide examples of optimal pure code families in this sense. The former employs a family of bivariate $J$-affine variety codes with the Euclidean inner product, while the latter uses BCH codes and homothetic-BCH codes (see \cite{hbch25} for the latter family) with both Euclidean and Hermitian inner products. We view all of these codes as evaluation codes. Homothetic-BCH codes can be understood as enlargements of BCH codes, which is how we introduce them. \Cref{subsec:BCH} focuses on the BCH code construction. We utilize \Cref{prop:qui}, which allows us to treat duals of BCH codes as classical $(r,\delta)$-LRCs. \Cref{prop:hulldim} determines the amount of entanglement required for the derived quantum code, and \Cref{pro:purity} serves to establish its purity. Theorems \ref{thm:qminusone} and \ref{thm:qplusone} present two families of optimal pure EA $(r,\delta)$-QLRCs arising from the Euclidean construction, whereas \Cref{thm:qplusonehermitian} provides a family derived from the Hermitian construction with the aid of Propositions \ref{prop:hdimo} and \ref{prop:hdime}, which evaluate the amount of entanglement via \Cref{prop:hulldim}.

Finally, \Cref{subsec:homo} addresses the homothetic-BCH code construction. In this framework, \Cref{cor:homoclasrd}, \Cref{prop:homohulldim}, \Cref{pro:purityhbch}, \Cref{thm:homoqminusone}, \Cref{thm:homoqplusone} and \Cref{thm:homoqplusonehermitian} serve as the direct counterparts (in terms of their respective roles) to \Cref{prop:qui}, \Cref{prop:hulldim}, \Cref{pro:purity}, \Cref{thm:qminusone}, \Cref{thm:qplusone} and \Cref{thm:qplusonehermitian}.
%\medskip

The paper is organized as follows. \Cref{sec:prelim} introduces notation and reviews background material on quantum stabilizer codes (with or without entanglement assistance) as well as classical and quantum $(r,\delta)$-LRCs. In \Cref{sec:EAQLRCs}, we formalize EA $(r,\delta)$-QLRCs, prove characterization theorems, and derive a Singleton-like bound for EA $(r,\delta)$-QLRCs obtained from Hermitian and Euclidean constructions. \Cref{sec:Jaf,sec:BCHhomo} provide constructions of optimal pure EA $(r,\delta)$-QLRCs using bivariate $J$-affine variety codes, BCH codes, and homothetic-BCH codes.

\section{Preliminaries}\label{sec:prelim}

\subsection{Notations}
\label{SS21}
Let $\mathbb{F}_q$ denote the finite field with cardinality $q=p^m$, where $p$ is a prime and $m$ a positive integer, and let $\mathbf{c}=(c_1,\dots,c_n)$, $\mathbf{c'}=(c'_1,\dots,c'_n)\in \mathbb{F}_q^n$, with $n$ a positive integer. In general, we will lighten notation by using bold letters to denote vectors of any length, where $c$, $2n$, $2c$ or $2(n+c)$ may be used in place of $n$. We denote the Euclidean inner product of $\mathbf{c}$ and $\mathbf{c'}$ as $\mathbf{c} \cdot_e \mathbf{c'}$ and recall that, for $\mathbf{c},\,\mathbf{c'}\in \mathbb{F}_{q^2}^n$, the Hermitian inner product is defined as $\mathbf{c} \cdot_h \mathbf{c'}=\mathbf{c} \cdot_e [\mathbf{c'}]^q$, where $[\mathbf{c'}]^q$ denotes the componentwise $q$-th power of $\mathbf{c'}$. Now, set $\mathbf{c}=(\mathbf{c}^a \mid \mathbf{c}^b)$ and $\mathbf{c'}=(\mathbf{c'}^a \mid \mathbf{c'}^b)\in \mathbb{F}_q^{2n}$, where $\mathbf{c}^a,\,\mathbf{c}^b,\,\mathbf{c'}^a,\,\mathbf{c'}^b \in \mathbb{F}_q^n$ (we mainly use this notation for vectors of even length). In this paper, we work with the \textit{symplectic form} defined as
$$(\mathbf{c}^a \mid \mathbf{c}^b) \cdot_s (\mathbf{c'}^a \mid \mathbf{c'}^b) := \mathbf{c}^a \cdot_e \mathbf{c'}^b - \mathbf{c'}^a \cdot_e \mathbf{c}^b.$$

We use error-correcting codes $C$ (either in $\mathbb{F}_q^n$, $\mathbb{F}_{q^2}^n$ or $\mathbb{F}_q^{2n}$), where $C^{\perp_e}$ (respectively, $C^{\perp_h},\,C^{\perp_s}$) denotes the dual code with respect to the Euclidean (respectively, Hermitian, symplectic) inner product. When using linear codes in $\mathbb{F}_q^n$, parameters are denoted $[n,k,d]_q$. The dimension of a code $C$ is denoted $\dim(C)$. The linear span of a set $W$ is denoted $\mathrm{span}\,W$.

In the next paragraph, we introduce notation for shortened and punctured codes restricted to a set $R\subseteq\{1,\dots,n\}$ of coordinates that are kept, with cardinality $\card{R}=r$. To do so, we consider the projection map onto the coordinates of $R$,
$$\pi_R \colon \mathbb{F}_q^n \to \mathbb{F}_q^r, \quad \mathbf{c} \mapsto \mathbf{c}_R := (c_i)_{i\in R},$$
and denote $\mathrm{supp}(\mathbf{c})=\{i\in\{1,\dots,n\} : c_i \neq 0\}$. The following definitions will also be used for codes $C \subseteq \mathbb{F}_q^{2n}$, where $\pi_R$ instead denotes the map
$$\pi_R \colon \mathbb{F}_q^{2n} \to \mathbb{F}_q^{2r}, \quad (\mathbf{c}^a \mid \mathbf{c}^b) \mapsto (\mathbf{c}^a \mid \mathbf{c}^b)_R := (\mathbf{c}^a_R \mid \mathbf{c}^b_R),$$
and $\mathrm{supp}(\mathbf{c}^a \mid \mathbf{c}^b)=\{i\in\{1,\dots,n\} : (c_i^a, c_i^b) \neq (0,0)\}$.

The \textit{punctured code of $C$ at $R$} is the code
$$\pi_R(C)=\{\mathbf{c}_R : \mathbf{c} \in C\},$$
and the \textit{shortened code of $C$ on $R$} is the code
$$\sigma_R(C)=\{\mathbf{c}_R : \mathbf{c} \in C \text{ with } \mathrm{supp}(\mathbf{c})\subseteq R\}.$$

The Hamming weight of an element $\mathbf{c} \in \mathbb{F}_q^n$ is denoted $\mathrm{wt}(\mathbf{c})=$ $\card{(\mathrm{supp}(\mathbf{c}))}$. The symplectic weight of an element $(\mathbf{c}^a\mid\mathbf{c}^b) \in \mathbb{F}_q^{2n}$ is $\mathrm{swt}(\mathbf{c}^a\mid\mathbf{c}^b) := \card{(\mathrm{supp}(\mathbf{c}^a\mid\mathbf{c}^b))}.$ Moreover, we use the notation $d_H(C):=\min\{\mathrm{wt}(\mathbf{c}) : \mathbf{0} \neq \mathbf{c}\in C\}$ for any nonempty set $C\subseteq \mathbb{F}_q^{n}$ with $C \neq \{\mathbf{0}\}$ and $d_s(C):=\min\{\mathrm{swt}(\mathbf{c}) : \mathbf{0} \neq \mathbf{c}\in C\}$ for any nonempty set $C\subseteq \mathbb{F}_q^{2n}$  with $C \neq \{\mathbf{0}\}$. For the zero linear space $\{\mathbf{0}\}$ and the empty set $\emptyset$, we define $d_H$ and $d_s$ to be $\infty$.

\subsection{Quantum stabilizer codes}

We briefly recall the definition of quantum stabilizer codes as in \cite{AK}. Let $\mathbb{C}^*$ be the multiplicative group of non-zero complex numbers, $\iota = \sqrt{-1}$, $\xi = e^{\frac{\iota 2 \pi}{p}}$ a primitive $p$-th root of unity, and consider the Hilbert space $\mathbb{C}^{q^n} = \mathbb{C}^q \otimes \cdots \otimes \mathbb{C}^q$. Given $(\mathbf{a}\mid\mathbf{b}) \in \mathbb{F}_q^{2n}$, $E_{(\mathbf{a},\mathbf{b})}$ denotes an error operator constructed by tensoring generalized Pauli matrices determined by $\mathbf{a}$ and $\mathbf{b}$. We denote by $G_n$ the error group generated by the nice error basis $\epsilon_n := \left\{ E_{(\mathbf{a},\mathbf{b})} : (\mathbf{a}\mid\mathbf{b}) \in \mathbb{F}_q^{2n} \right\}$ on $\mathbb{C}^{q^n}$.

%Let $\C^*$ be the multiplicative group of complex numbers, $\iota = \sqrt{-1}$ and $\xi = e^{\frac{\iota 2 \pi}{p}}$ a primitive $p$-th root of unity. The Hilbert space is $\C^{q^n} = \C^q \otimes \cdots \otimes \C^q$\textcolor{blue}{, and we denote by $\{\ket{x} : x\in \F_q\}$ an orthonormal basis of $\C^q$}. Given $(\mathbf{a}|\mathbf{b}) \in \F_q^{2n}$, $E_{(\mathbf{a},\mathbf{b})}$ denotes an error operator constructed by tensoring generalized Pauli matrices $X(a)$ and $Z(b)$ acting on $\C^q${\color{blue}:
%$$X(a) : \C^q \to \C^q, \quad X(a)\ket{x}=\ket{x+a},$$
%and
%$$Z(b) : \C^q \to \C^q, \quad Z(b)\ket{x}=\xi^{\tr_{q/p}(bx)}\ket{x},$$
%where $\tr_{q/p}$ is the trace map $\tr_{q/p} \colon \F_q \to \F_p$, $\tr_{q/p}(x)=x+x^p+\cdots+x^{p^{m-1}}$.}
%Consider a nice error basis $\epsilon_n := \{ E_{(\mathbf{a},\mathbf{b})} : (\mathbf{a}|\mathbf{b}) \in \F_q^{2n}\}$ on $\C^{q^n}$ and the error group associated with $\epsilon_n$:
%\[
%G_n := \{\xi^\ell E_{(\mathbf{a},\mathbf{b})} \; : \; (\mathbf{a}|\mathbf{b}) \in \F_q^{2n},\, 0 \leq \ell \leq p-1\}.
%\]
%\textcolor{blue}{Let $S\subseteq G_n$ be an abelian subgroup, with identity $I$, containing the center $\{ \xi^\ell I : 0 \leq \ell \leq p-1\}$. Consider also a linear character $\lambda$ of $S$, i.e., a group homomorphism $\lambda \colon S \to \C^*$, such that $\lambda (\xi I) = \xi$.}

\begin{defi}
Let $S\subseteq G_n$ be an abelian subgroup containing the center and let $\lambda$ be a linear character of $S$ such that $\lambda (\xi \mathcal{I}) = \xi$, $\mathcal{I}$ being the identity. A \textit{(quantum) stabilizer code} $Q$, associated with $S$ and $\lambda$, is the common eigenspace
	\[
		Q = Q(S) := \bigcap_{E \in S} \left\{\mathbf{v} \in \mathbb{C}^{q^n} \; : \; E \mathbf{v} = \lambda(E) \mathbf{v} \right\}.
	\]
\end{defi}

Let $E = \xi^\ell E_{(\mathbf{a},\mathbf{b})} \in G_n$ be an error; the weight of $E$ is defined as $\mathrm{wt}(E) = \mathrm{swt}(\mathbf{a}\mid\mathbf{b})$. We say that a stabilizer code $Q$ has minimum distance $d$ whenever $d$ is the minimum weight of an error in $G_n$ that changes a codeword in $Q$ without being detected by $Q$. A stabilizer code $Q \subseteq \mathbb{C}^{q^n}$ of dimension $K$ and minimum distance $d$ is said to be an $((n,K, d))_q$ code. If $K=q^k$, we write $[[n, k, d]]_q$.

A stabilizer code $Q(S)$ induces an additive code
$$
C := \left\{ (\mathbf{a}| \varphi^{-1} (\mathbf{b})) \; : \; \xi^\ell E_{(\mathbf{a},\mathbf{b})} \in S \right\} \subseteq \F_q^{2n},
$$
where $\varphi$ is a field automorphism of $\F_q$ applied componentwise. This automorphism %(that we usually omit)
guarantees that $C$ is self-orthogonal with respect to the following \textit{trace-symplectic form}:
$$(\mathbf{a} | \mathbf{b}) \cdot_{\textrm{trs}} \left(\mathbf{a'} | \mathbf{b'}\right)=\textrm{tr}_{q/p} \left(\mathbf{a} \cdot_e \mathbf{b'}-\mathbf{a'} \cdot_e \mathbf{b}\right),$$
where $\tr$ is the trace map $\tr_{q/p} \colon \F_q \to \F_p$.
%\textcolor{blue}{\sout{where $\tr$ is the trace map $\tr_{q/p} \colon \F_q \to \F_p$, $\tr_{q/p}(x)=x+x^p+\cdots+x^{p^{m-1}}$.}}

The connection between additive codes over $\F_q$ and stabilizer codes is provided in the next theorem, see \cite{Ketkar}.

\begin{thm}\label{thm:equiv}
	The existence of an $((n,K,d))_q$ stabilizer code is equivalent to that of an additive code $C \subseteq \F_q^{2n}$ with cardinality $q^n/K$, such that $C \subseteq C^{\perp_{\trs}}$ and, whenever $K>1$, $d$ is the minimum symplectic weight of the elements in  $C^{\perp_{\trs}} \setminus C$. When $K=1$, there is no undetectable error.
\end{thm}

Following \cite{Ketkar}, other construction methods are provided as consequences of the above result.
%Let $\{\alpha, \alpha^q\}$ be a normal basis of $\F_{q^2}$ over $\F_{q}$. Using the \textit{trace-alternating form} $\cdot_a$:
%$$\mathbf{a} \cdot_a \mathbf{b} := \mathrm{tr}_{q/p} \left( \frac{ \mathbf{a}\cdot_e \mathbf{b}^q - \mathbf{a}^q \cdot_e \mathbf{b}}{\alpha^{2q}-\alpha^2 }\right), \quad \mathbf{a},\,\mathbf{b} \in \F_{q^2}^n$$
%and with the bijective isometry $\psi: \F_{q}^{2n} \rightarrow \F_{q^2}^n$ defined by $\psi(\mathbf{a}| \mathbf{b}) = \alpha \mathbf{a} + \alpha^q \mathbf{b}$, the following relation between the forms $\cdot_{\trs}$ and $\cdot_a$ holds: $(\mathbf{a}| \mathbf{b}) ._{\mathrm{trs}} (\mathbf{a'}| \mathbf{b'}) = \psi\left((\mathbf{a}| \mathbf{b})\right) \cdot_a \psi \left((\mathbf{a'}| \mathbf{b'})\right)$. Then, \Cref{thm:equiv} can be traslated to state that the existence of a quantum stabilizer code is equivalent to that of an additive code $C \subseteq \F_{q^2}^n$ satisfying $C \subseteq C^{\perp_a}$, where $C^{\perp_a}$ denotes the dual code of $C$ with respect to $\cdot_a$.
We are interested in the construction of stabilizer codes from the Hermitian and Euclidean inner products. These results are summarized as follows, see \cite{Ketkar}.

\begin{thm}
	\label{resto_en}
	The following statements hold:
	\begin{enumerate}
		\item Let $C$ be a linear $[n,(n-k)/2,d]_{q^2}$ code such that $C \subseteq C^{\perp_h}$. Denote by $d^{\perp_h}$ the minimum distance of $C^{\perp_h}$. Then, there exists an $[[n,k,\geq d^{\perp_h}]]_q$ stabilizer code.
		%\item \textcolor{blue}{\sout{Let $C_1$ and $C_2$ be linear codes with parameters $[n,k_i,d_i]_q$, $1 \leq i \leq 2$, such that $C_2^{\perp_e} \subseteq C_1$. Then, there exists an $[[n, k_1+k_2-n,d]]_q$ stabilizer code whose minimum distance is $d= \min \{\wt(\mathbf{c}) : \mathbf{c} \in (C_1 \setminus C_2^{\perp_e}) \cup (C_2 \setminus C_1^{\perp_e})\}$.}}
		\item Let $C$ be a linear $[n,(n-k)/2,d]_q$ code such that $C \subseteq C^{\perp_e}$. Denote by $d^{\perp_e}$ the minimum distance of $C^{\perp_e}$. Then, there exists an $[[n,k,\geq d^{\perp_e}]]_q$ stabilizer code.
	\end{enumerate}
\end{thm}

Although our results are stated using the terminology of \cite{Ketkar}, the definition of stabilizer codes we employ comes from \cite{AK}. The approach in \cite{AK} initially establishes results over $\mathbb{F}_p$ before generalizing them to $\mathbb{F}_q$. We apply this same extension procedure here; details of this procedure can be found in \cite[Section 3.1]{galindo19} within the framework of entanglement-assisted codes (defined below). Therefore, in this paper we consider the symplectic form defined in Subsection \ref{SS21}; that is, we do not employ $\mathbb{F}_p$-linear (additive) codes in $\mathbb{F}_q^{2n}$ or trace inner products.

To simplify the theory of quantum error correction and achieve higher communication capacity, the authors in \cite{Brun} proposed sharing entanglement between the encoder and decoder, leading to the following family of quantum codes.

\begin{defi}
	Let $C \subseteq \mathbb{F}_q^{2n}$ be any linear code (without necessarily satisfying inclusion requirements as above). An \textit{entanglement-assisted quantum error-correcting code (EAQECC)}, $Q'(C)$, constructed from $C$, is a QECC whose encoding quantum circuit is constructed from another classical code $C' \subseteq \mathbb{F}_q^{2(n+c)}$ such that $C' \subseteq C'^{\perp_s}$ and $\pi_{\{1,\dots,n\}}(C')=C$, with $c$ being the minimum required number of maximally entangled pairs (quantum states in $\mathbb{C}^q \otimes \mathbb{C}^q$) to encode logical qudits into physical qudits. The decoder uses the stabilizer defined by $C'$ for decoding, by applying a decoder for $Q'(C')$ to the concatenation of the decoder's entanglement and the received word.
\end{defi}

For a fixed $C$, there may be multiple choices for such a code $C'$. The existence of a code $C'$ as above was proved in \cite{Brun} for the binary case, in \cite{luo17} for prime $q$, and in \cite{galindo19} for the general case. In that general setting of $q$-ary EAQECCs, explicit formulas for the parameters $\left[[n,\kappa,d;c]\right]_q$ can be found in \cite{galindo19}, using the symplectic, Hermitian, or Euclidean constructions. A different choice of $C'$ does not modify the parameters. We recall them in the next theorem.

\begin{thm}\label{thm:csh}
	Let $C \subseteq \mathbb{F}_q^{2n}$ (respectively, $C\subseteq \mathbb{F}_{q^2}^n$ and $C \subseteq \mathbb{F}_q^n$) be a $[2n,n-k]_q$ (respectively, $[n,(n-k)/2]_{q^2}$ and $[n, (n-k)/2]_q$) linear code. Then,  the minimum required number of maximally entangled pairs $c$ of an EAQECC $Q'(C)$ coming from $C$ can be computed as follows:
	$$2c = \dim C - \dim \left(C \cap C^{\perp_s}\right)$$
	(respectively,
	$$c = \dim C - \dim \left(C \cap C^{\perp_h}\right)\textrm{, and}$$
	$$c = \dim C - \dim \left(C \cap C^{\perp_e}\right)\textrm{.})$$
	
The EAQECC has parameters $[[n, k + c, d; c]]_q$, where
	$$d = d_s\left(C^{\perp_s} \setminus \left(C \cap C^{\perp_s}\right)\right)$$
	(respectively,
	$$d = d_H\left(C^{\perp_h} \setminus \left(C \cap C^{\perp_h}\right)\right)\textrm{, and}$$
	$$d = d_H\left(C^{\perp_e} \setminus \left(C \cap C^{\perp_e}\right)\right)\textrm{).}$$
\end{thm}

We end this subsection by stating that, in the rest of the paper, $Q'(C)$ stands for the QECC (respectively, EAQECC) defined by a self-orthogonal code $C$ (respectively, a self-orthogonal code $C'$ such that $\pi_{\{1,\dots,n\}}(C')=C$).

\subsection{Classical and quantum $(r,\delta)$-locally recoverable codes}\label{sub:cqlrcs}

Let us recall the classical definition, given in \cite{PKLK2012}, which is designed to correct simultaneous erasures from a few other coordinates.

\begin{defi}
	A code $C \subseteq \mathbb{F}_q^n$ is an \textit{$(r,\delta)$-LRC} (or an \textit{LRC with locality $(r,\delta)$}) if for each position $i \in \{1, \ldots, n\}$ there exists a set of positions $K \subseteq \{1, \ldots, n\}$, called an $(r,\delta)$-\textit{recovery set} for $i$, such that:
	\begin{enumerate}
		\item $i\in K$, $\mathrm{card}\,(K) \leq r + \delta -1$; and
		\item $d_H[\pi_{K} (C)] \geq \delta$.
	\end{enumerate}
	With the above definition, one is able to correct an erasure at coordinate $i$ plus any other $\delta-2$ erasures in $K \setminus \{i\}$ by using at most $r$ remaining coordinates.
\end{defi}

There is a Singleton-like bound for $(r,\delta)$-LRCs (see \cite{PKLK2012}), stating that if $[n,k,d]_q$ are the parameters of an $(r,\delta)$-LRC, then
	\begin{equation}
		\label{Singleform}
		k+d+ \left(\left\lceil{\frac{k}{r}} \right\rceil -1 \right)\left( \delta -1 \right) \leq n+1.
	\end{equation}
If the above bound is attained, the code is said to be \textit{optimal} ($(r,\delta)$-LRC).

\medskip

The concept of a \textit{quantum $(r,\delta)$-LRC} was introduced in \cite{qlrc24}, which we follow. Let $Q \subseteq \mathbb{C}^{q^n}$ be a quantum code and $\emptyset \neq I \subsetneq K \subseteq \{1,\dots,n\}$ be sets of positions, where $I$ denotes the ones where erasures occur, and $K$ denotes the ones used for recovery. With the aim of reversing the action of a completely depolarizing channel by using qudits at $K$, the concept of a quantum $(I,K)$-LRC was first introduced in \cite{qlrc24}. Let $\Gamma$ be the completely depolarizing channel on a single qudit in $\mathbb{C}^q$ and set $\Gamma^I = \bigotimes_{i=1}^n \Gamma_i$, where $\Gamma_i = \Gamma$ for $i \in I$ and $\Gamma_i$ is the identity for the remaining $i$'s.

The code $Q$ is said to be an \textit{$(I,K)$-locally recoverable code} if there exists a trace-preserving quantum operation (in the sense of \cite[Chapter 8]{chuangnielsen}) $\mathcal{R}^{K}_{Q,I}$, acting only on the qudits corresponding to $K$ and keeping the remaining ones untouched, such that
\begin{equation*}
	\mathcal{R}^{K}_{Q,I} \circ \Gamma^I (\ket{\varphi}\bra{\varphi}) = \ket{\varphi}\bra{\varphi} %\label{eq1_en}
\end{equation*}
for all $\ket{\varphi} \in Q$.

\begin{defi}
\label{CC-2}
	A quantum code $Q \subseteq \mathbb{C}^{q^n}$ is said to be a \textit{quantum $(r,\delta)$-locally recoverable code} if for each index $i \in \{1,\dots,n\}$, there exists a set $K \subseteq \{1,\dots,n\}$ such that $i \in K$, $\mathrm{card}(K) \leq r+\delta-1$, and for every subset $I\subseteq K$ of cardinality $\delta-1$, $Q$ is $(I,K)$-locally recoverable.
\end{defi}

There exists a necessary and sufficient condition on the classical code defining a stabilizer code for it to be a quantum $(I,K)$-LRC (see \cite{qlrc24}). We gather the symplectic, Hermitian, and Euclidean conditions in the following result.

\begin{thm}\label{thm:equivcond}
	Let $Q'(C) \subseteq \mathbb{C}^{q^n}$ be a stabilizer code given by a $q$-ary symplectic (respectively, $q^2$-ary linear Hermitian and $q$-ary linear Euclidean) self-orthogonal code $C \subseteq \mathbb{F}_q^{2n}$ (respectively, $ C \subseteq \mathbb{F}_{q^2}^n$ and  $C \subseteq \mathbb{F}_{q}^n$). Consider sets of positions $I$ and $K$ such that $\emptyset \neq I \subsetneq K \subseteq \{1, \ldots, n\}$. Then, $Q'(C)$ is $(I,K)$-locally recoverable if and only if the equality
	\begin{equation}\label{eq:equivcond}
		\sigma_I[\pi_{K}(C^{\perp})] = \sigma_I(C)
	\end{equation}
	holds, where $\perp$ denotes $\perp_s$ (respectively, $\perp_h$ and $\perp_e$).
\end{thm}

In the language of quantum $(r,\delta)$-LRCs, the above result translates as follows.

\begin{thm}
	\label{quantumH-E}
	Let $Q'(C)$ and $C$ be codes as in the statement of \Cref{thm:equivcond}. Then, $Q'(C)$ is a quantum $(r,\delta)$-LRC if and only if, for every position $i \in \{1, \ldots, n\}$, there exists a set $K \ni i$ with $\mathrm{card}(K) \leq r + \delta -1$ such that Equality \eqref{eq:equivcond} holds for any subset $I \subsetneq K$ with $\mathrm{card}(I)=\delta -1$.
\end{thm}

In addition, assuming Hermitian or Euclidean dual-containment of the classical codes $C$ defining quantum codes, we set $Q(C):=Q'(C^\perp)$. A relation between classical and quantum local recoverability was established in \cite{qlrc24}. Indeed, under the condition $\mathrm{card}(I) \leq d(C^{\perp} ) -1$, $Q(C)$ being a quantum $(I,K)$-LRC is equivalent to $C$ being able to correct erasures in $I$ from coordinates of $C$ at $K \setminus I$.

The result in terms of $(r,\delta)$-local recoverability is as follows.

\begin{thm}
	\label{thm:rdeltarel}
	Let $C$ be a $q^2$-ary linear Hermitian (respectively, $q$-ary linear Euclidean) dual-containing code $C \subseteq \mathbb{F}_{q^2}^n$ (respectively, $\mathbb{F}_{q}^n$). Then, $Q(C)$ is a quantum $(r,\delta)$-LRC if $C$ is a classical $(r,\delta)$-LRC.
The last implication becomes a logical equivalence if $\delta \leq d(C^{\perp_h})$ (respectively, $d(C^{\perp_e})$).
\end{thm}

To conclude, we provide a Singleton-like bound for quantum $(r,\delta)$-LRCs coming from Hermitian or Euclidean dual-containing classical codes (see \cite{qlrc24}).

\begin{thm}
	\label{thm:SingletonQ}
	Let $C$ be a $q^2$-ary linear Hermitian (respectively, $q$-ary linear Euclidean) dual-containing code $C \subseteq \mathbb{F}_{q^2}^n$ (respectively, $\mathbb{F}_{q}^n$) such that $\dim C = \frac{n+k}{2}$ and $C$ is an $(r,\delta)$-LRC. Then, $Q(C)$ is a quantum $(r,\delta)$-LRC and it has parameters $[[n,k >0, \geq d(C)]]_q$, which satisfy
	\begin{equation}
		\label{eq:SingletonQ}
		k + 2 d(C) + 2\left(\left\lceil\frac{n+k}{2r}\right\rceil-1\right)(\delta-1) \leq n+2.
	\end{equation}
\end{thm}

A quantum code with parameters $[[n,k,d]]_q$ coming from a dual-containing classical code $C$ as in Theorem \ref{thm:SingletonQ} is \textit{pure} if and only if $d=d(C)$.

\begin{defi}
A \textit{pure} quantum $(r, \delta)$-LRC $Q(C)$ as above which attains the bound in \eqref{eq:SingletonQ} is said to be \textit{optimal}.
\end{defi}

\section{Entanglement-assisted quantum $(r,\delta)$-LRCs}\label{sec:EAQLRCs}

In this section, we introduce entanglement-assisted quantum $(r,\delta)$-LRCs and provide a necessary and sufficient condition for an EAQECC to be a quantum $(r,\delta)$-LRC. In the case of Hermitian and Euclidean constructions, we also relate classical and entanglement-assisted quantum $(r,\delta)$-local recoverability, and we obtain a Singleton-like bound for entanglement-assisted codes coming from these constructions. Note that, very recently, EAQECCs for a single erasure obtained from the CSS construction were introduced in \cite{EAQLR-r}, and for the general case with availability in \cite{EAQLR-gret}.

Let $C \subseteq \F_q^{2n}$ and $C' \subseteq \F_q^{2(n+c)}$ be linear codes such that $C' \subseteq C'^{\perp_s}$ and $\pi_{\{1,\dots,n\}}(C')=C$. Let $Q'(C)$ be the associated EAQECC, where $c$ is given by \Cref{thm:csh}. Note that, for a fixed $C$, there are different choices for $C'$. To define entanglement-assisted quantum $(r,\delta)$-LRCs, we wish to provide a condition independent of the choice of $C'$ and dependent only on $C$. We assume that, on the $c$ pairs of entangled qudits, the qudits held by the receiver are unaffected by erasures.

%Let $\emptyset \neq I \subsetneq K \subseteq \{1,\dots,n\}$ be sets of positions. Let $\Gamma$ be the completely depolarizing channel on a single qudit in $\C^q$ and set $\Gamma^I = \bigotimes_{i=1}^{n+c} \Gamma_i$, where $\Gamma_i = \Gamma$ for $i \in I$ and $\Gamma_i$ is the identity for the remaining indices $i$.

\begin{defi}\label{def:eaIKlrc}
Let $\emptyset \neq I \subsetneq K \subseteq \{1,\dots,n\}$ be sets of positions. An EAQECC $Q'(C)$ given by $C \subseteq \F_q^{2n}$ as above is said to be \textit{$(I,K)$-locally recoverable} if there exists a code $C'$ as above and a trace-preserving quantum operation (in the sense of \cite[Chapter 8]{chuangnielsen}) $\mathcal{R}^{K \cup \{n+1,\dots,n+c\}}_{Q' (C'),I}$, acting only on the qudits corresponding to $K \cup \{n+1,\dots,n+c\}$ and keeping the remaining ones untouched, such that
	\begin{equation*}
		\mathcal{R}^{K \cup \{n+1,\dots,n+c\}}_{Q' (C'),I} \circ \Gamma^I (\ket{\varphi}\bra{\varphi}) = \ket{\varphi}\bra{\varphi} \label{eq1_en}
	\end{equation*}
for all $\ket{\varphi} \in Q'(C')$, where $\Gamma^I = \bigotimes_{i=1}^{n+c} \Gamma_i$, with $\Gamma_i$ being $\Gamma$ (as defined before Definition \ref{CC-2}) for $i \in I$, and $\Gamma_i$ being the identity for the remaining $i$'s.
\end{defi}

Note that the forthcoming Theorem \ref{thm:1} shows that Definition \ref{def:eaIKlrc} is independent of the choice of $C'$. Next we introduce the concept of entanglement-assisted quantum $(r,\delta)$-locally recoverable code.

\begin{defi}\label{def:eaqlrc}
An EAQECC $Q'(C)$ is said to be an \textit{entanglement-assisted quantum $(r,\delta)$-locally recoverable code} (also \textit{EA $(r,\delta)$-QLRC}) if for each index $i \in \{1,\dots,n\}$, there exists a set $K \subseteq \{1,\dots,n\}$ such that $i \in K$, $\mathrm{card}(K) \leq r+\delta-1$, and for every subset $I \subseteq K$ of cardinality $\delta-1$, $Q'(C)$ is $(I,K)$-locally recoverable.
\end{defi}

\begin{rem}\label{rem:assistedrdelta}
All definitions in Subsection \ref{sub:cqlrcs} remain the same in this context, taking into account that the quantum operation for local recovery is also assisted by the qudits held by the receiver. We choose this point of view to avoid introducing new definitions of $(r,\delta)$-local recoverability for EAQECCs, which would require considering larger sets of positions and alternative interpretations of the locality $(r,\delta)$. %$In that case we would consider a set $L$ such that J \subseteq L \subseteq \{1,\dots,n+c\}\}$, $L$ being the union of $J$ and the set of positions of entangled qudits. Then, define $(I,L)$-local recoverability...
\end{rem}

As announced, the following result characterizes the $(I,K)$-local recoverability of an EA quantum error-correcting code $Q'(C)$ by a condition depending only on $C$ and not on $C'$.

\begin{thm}\label{thm:1}
Let $C \subseteq \F_q^{2n}$ be a linear code and $Q'(C)$ be the EAQECC constructed from a linear code $C' \subseteq \F_q^{2(n+c)}$ such that $C' \subseteq C'^{\perp_s}$ and $\pi_{\{1,\dots,n\}}(C')=C$, where $c$ is given by \Cref{thm:csh}. Assume that the qudits held by the receiver are unaffected by erasures. Consider sets of positions $I$ and $K$ such that $\emptyset \neq I \subsetneq K \subseteq \{1, \ldots, n\}$. Then, $Q'(C)$ is $(I,K)$-locally recoverable if and only if the equality
\begin{equation}\label{eq:nequivcond}
	\sigma_I[\pi_{K}(C^{\perp_s})] = \sigma_I(C \cap C^{\perp_s})
\end{equation}
holds.	
\end{thm}

\begin{proof}
It suffices to prove the statement for a prime $q$, i.e., for the case of a prime field $\mathbb{F}_q=\mathbb{F}_p$, because the argument can be generalized to an arbitrary $\mathbb{F}_q$ ($q=p^m$) by applying the extension procedure described in \cite[Section 3.1]{galindo19}.

Consider sets $I$ and $K$ such that $\emptyset \neq I \subsetneq K \subseteq \{1, \ldots, n\}$. Let $L := K \cup \{n+1, \ldots, n+c\}$. Since the decoder uses the stabilizer defined by $C'$ for decoding, by applying a decoder for
$Q'(C')$ to the concatenation of the decoder's entanglement and the received word, and  by Theorem \ref{thm:equivcond}, $Q'(C)$ is $(I,K)$-locally recoverable if and only if
\begin{equation}\label{eq:equivL}
	\sigma_I[\pi_L(C'^{\perp_s})] = \sigma_I(C').
\end{equation}
Thus, the result will be proved if we show the following equivalence:
\begin{equation} \label{eq:equivalenceC}
    \sigma_I[\pi_L((C')^{\perp_s})] = \sigma_I(C') \iff \sigma_I[\pi_K(C^{\perp_s})] = \sigma_I(C \cap C^{\perp_s}).
\end{equation}
We prove Equivalence \eqref{eq:equivalenceC} by showing (1) $\sigma_I(C') = \sigma_I(C \cap C^{\perp_s})$ and (2) $\sigma_I[\pi_L((C')^{\perp_s})] = \sigma_I[\pi_K(C^{\perp_s})]$.

Keeping the notation in Subsection \ref{SS21}, denote by $\pi_n = \pi_{\{1, \ldots, n\}}$ the projection onto the physical qudits, and $\pi_e = \pi_{\{n+1, \ldots, n+c\}}$ the projection onto the extra entanglement qudits. As a first step, it is convenient to observe that \textit{the projection $\pi_n: C' \to C$ is an isomorphism.} Let us show this. It is clear that $\pi_n(C') = C$. Now, any two elements $\mathbf{x}$ and $\mathbf{z}$ in $C'$ satisfy $\mathbf{x} \cdot_s \mathbf{z} = 0$, which yields $\pi_n(\mathbf{x}) \cdot_s \pi_n(\mathbf{z}) + \pi_e(\mathbf{x}) \cdot_s \pi_e(\mathbf{z}) = 0$, and thus,
\begin{equation}
\label{CA1}
\pi_e(\mathbf{x}) \cdot_s \pi_e(\mathbf{z}) = - \pi_n(\mathbf{x}) \cdot_s \pi_n(\mathbf{z}).
\end{equation}
Denote by $C'_e$ the projection $\pi_e(C') \subseteq \mathbb{F}_p^{2c}$. The quotient space $C / (C \cap C^{\perp_s})$ has a non-degenerate symplectic form and dimension $2c$. Equality \eqref{CA1} shows that the symplectic form on $C'_e$ matches the non-deg\-ene\-rate form on $C / (C \cap C^{\perp_s})$, and therefore $C'_e$ is a non-degenerate symplectic space of dimension $2c$. Since $C'_e \subseteq \mathbb{F}_p^{2c}$, it must hold that $C'_e = \mathbb{F}_p^{2c}$. Now, consider the kernel of the projection $\pi_n$, $\ker(\pi_n)$, and let $\boldsymbol{\theta} = (\mathbf{0}, \mathbf{y}) \in \ker(\pi_n)$, which must be symplectically orthogonal to every vector $\mathbf{z} = (\mathbf{z}_n, \mathbf{z}_e) \in C'$. Then,
\[ \boldsymbol{\theta} \cdot_s \mathbf{z} = (\mathbf{0}, \mathbf{y}) \cdot_s (\mathbf{z}_n, \mathbf{z}_e) = \mathbf{y} \cdot_s \mathbf{z}_e = 0. \]
Finally, since $\mathbf{y}$ is orthogonal to all of $C'_e = \mathbb{F}_p^{2c}$, it holds that $\mathbf{y}= \mathbf{0}$, which proves that $\ker(\pi_n)$ vanishes and $\pi_n: C' \to C$ is an isomorphism.

Our second step is to prove (1). We start by proving that the unique lift $(\mathbf{w}, \mathbf{e}_w) \in C'$ of an element $\mathbf{w} \in C \cap C^{\perp_s}$ is $(\mathbf{w}, \mathbf{0})$. Indeed,
\[ (\mathbf{w}, \mathbf{e}_w) \cdot_s (\mathbf{v}, \mathbf{f}_v) = \mathbf{w} \cdot_s \mathbf{v} + \mathbf{e}_w \cdot_s \mathbf{f}_v = 0 \]
for any lift $(\mathbf{v}, \mathbf{f}_v) \in C'$. Therefore, $\mathbf{w} \in C^{\perp_s}$ forces $\mathbf{e}_w \cdot_s \mathbf{f}_v = 0$ for all $\mathbf{f}_v \in C'_e$, and thus $\mathbf{e}_w = \mathbf{0}$ since $C'_e = \mathbb{F}_p^{2c}$. When we apply $\sigma_I$ to $C'$, we restrict to vectors whose support lies strictly in $I \subseteq \{1, \ldots, n\}$. Any such vector $\mathbf{z} \in C'$ takes the form $(\mathbf{v}, \mathbf{0})$, where $\mathrm{supp}(\mathbf{v}) \subseteq I$. As we have shown, $(\mathbf{v}, \mathbf{0}) \in C'$ if and only if $\mathbf{v} \in C \cap C^{\perp_s}$. Consequently,
\[ \sigma_I(C') = \sigma_I(C \cap C^{\perp_s}). \]

Finally, we prove (2), which concludes the proof. Let $\mathbf{u}$ be an element
%of the set on the left-hand side.
to be shortened. By definition, $\mathbf{u}$ is the projection of some vector $\mathbf{c} \in (C')^{\perp_s}$ onto $L = K \cup \{n+1, \ldots, n+c\}$, subject to the constraint that $\mathrm{supp}(\mathbf{u}) \subseteq I$. The fact that $\mathrm{supp}(\mathbf{u}) \subseteq I$ forces the extra coordinates of $\mathbf{c}$ to be zero. Thus, $\mathbf{c} = (\mathbf{w}, \mathbf{0})$, where $\mathbf{w} \in \mathbb{F}_p^{2n}$ and the projection onto $K$ satisfies $\mathrm{supp}(\mathbf{w}_K) \subseteq I$. Reasoning as above, $\mathbf{c} \in (C')^{\perp_s}$ implies $\mathbf{w} \cdot_s \mathbf{v} = 0$ for all $\mathbf{v} \in C$, which means that $\mathbf{w} \in C^{\perp_s}$. Projecting $\mathbf{c}$ onto $L$ leaves $(\mathbf{w}_K, \mathbf{0})$, and restricting to $I$ yields exactly $\mathbf{w}_I$. The converse is clear because the set $\sigma_I[\pi_K(C^{\perp_s})]$ is constructed by taking $\mathbf{w} \in C^{\perp_s}$, projecting onto $K$, and filtering for vectors where $\mathrm{supp}(\mathbf{w}_K) \subseteq I$, yielding $\mathbf{w}_I$.

\end{proof}

\begin{rem}
Theorems 1 and 2 in \cite{EAQLR-r} and Proposition 3.4 in \cite{EAQLR-gret}
can be immediately deduced from \Cref{thm:1}.
\end{rem}

From the fact that quantum stabilizer codes can be constructed using Hermitian or Euclidean inner products instead of symplectic forms, \Cref{thm:1} can be straightforwardly translated as follows.

\begin{cor}\label{cor:1}
		Let $C \subseteq \F_{q^2}^n$ (respectively, $C \subseteq \F_q^n$) be a linear code and $Q'(C)$ be the EAQECC constructed from a linear code $C' \subseteq \F_{q^2}^{n+c}$ (respectively, $C' \subseteq \F_q^{n+c}$) such that $C' \subseteq C'^{\perp_h}$ (respectively, $C' \subseteq C'^{\perp_e}$) and $\pi_{\{1,\dots,n\}}(C')=C$, where $c$ is given by \Cref{thm:csh}. Assume that the qudits held by the receiver are unaffected by erasures. Consider sets of positions $I$ and $K$ such that $\emptyset \neq I \subsetneq K \subseteq \{1, \ldots, n\}$. Then, $Q'(C)$ is $(I,K)$-locally recoverable if and only if the equality
		\begin{equation}\label{eq:nequivcondHE}
			\sigma_I[\pi_K(C^\perp)] = \sigma_I(C \cap C^\perp)
		\end{equation}
		holds, where $\perp$ denotes $\perp_h$ (respectively, $\perp_e$).
\end{cor}

Moreover, \Cref{thm:1} (or \Cref{cor:1}) allows us to convert any classical $(r,\delta)$-LRC into a quantum one. Let $\perp$ denote $\perp_h$ or $\perp_e$ according to the case considered and, with the notation as before \Cref{thm:rdeltarel}, consider the EAQECC $Q(C)$ given by a classical code $C \subseteq \mathbb{F}^n_{q^2}$ (Hermitian case) or $C \subseteq \mathbb{F}^n_q$ (Euclidean case). That is, $Q(C)$ is the EAQECC given by a suitable dual-containing code $C' \subseteq \mathbb{F}^{n+c}_{q^2}$ in the Hermitian case and $C' \subseteq \mathbb{F}^{n+c}_q$ in the Euclidean case. Note that $C^\perp$ corresponds to $C$ in \Cref{cor:1}. In this case, $Q(C)$ is an $[[n,2\dim C-n+c,d_H(C \setminus (C \cap C^\perp));c]]_q$ EAQECC.

\begin{cor}\label{coro1}
	Let $C$ be a linear code as above and assume that $C$ is a classical $(r, \delta)$-LRC. Then, $Q(C)$ is an entanglement-assisted quantum $(r,\delta)$-locally recoverable code.
\end{cor}
\begin{proof}
Let $I$ and $K$ be, respectively, erasure and recovery sets such that $\emptyset \neq I \subsetneq K \subseteq \{1, \ldots, n\}$, with $\mathrm{card}(I) = \delta -1$ and $\mathrm{card}(K) \leq r+ \mathrm{card}(I) = r+\delta-1$.
Since $C$ is a classical $(r,\delta)$-LRC, we have $d_H(C \cap C^{\perp}) \geq d_H(C) \geq \delta$, and the right-hand side of \Cref{eq:nequivcondHE} (with $C$ corresponding to $C^\perp$ in the equation) equals $\{\mathbf{0}\}$, where $\mathbf{0}$ denotes the $n$-dimensional zero vector.
Again, since $C$ is a classical $(r,\delta)$-LRC, the left-hand side of \Cref{eq:nequivcondHE} also equals $\{\mathbf{0}\}$ and $Q(C)$ is $(I,K)$-locally recoverable, which completes the proof.
\end{proof}

The equivalence between classical and quantum $(r,\delta)$-local recoverability is derived from \Cref{coro1} and \cite[Proposition 27]{qlrc24}. We state the result below.

\begin{cor}\label{coro2}
	Let $C \subseteq \F_{q^2}^n$ (respectively, $C \subseteq \F_{q}^n$) be a linear code as above such that $d_H(C \cap C^{\perp}) \geq \delta$, where $\perp$ denotes $\perp_h$ (respectively, $\perp_e$). Then $Q(C)$ is an entanglement-assisted quantum $(r,\delta)$-locally recoverable code if and only if $C$ is a classical $(r, \delta)$-LRC.
\end{cor}

Similarly to \cite[Theorem 31, Remark 32]{qlrc24}, one obtains a Singleton-like bound for entanglement-assisted quantum $(r,\delta)$-LRCs coming from the Hermitian or Euclidean constructions.

\begin{cor}\label{coro3}
	Let $C \subseteq \F_{q^2}^n$ (respectively, $C \subseteq \F_q^n$) be a $q^2$-ary (respectively, $q$-ary) linear code as above defining an EAQECC, $Q(C)$, with parameters $[[n, k > 0, \geq d_H(C); c]]_q$, where $c=\dim C^\perp - \dim (C \cap C^\perp)$ and $\perp$ denotes $\perp_h$ (respectively, $\perp_e$). Assume also that either
	\begin{enumerate}
		\item $C$ is a classical $(r,\delta)$-LRC; or
		\item\label{assumption2} $Q(C)$ is an entanglement-assisted quantum $(r,\delta)$-LRC and $d_H(C \cap C^{\perp}) \geq \delta$.
	\end{enumerate}
 	Then, the entanglement-assisted quantum $(r,\delta)$-LRC, $Q(C)$, satisfies
	\begin{equation}\label{eq:EASing}
		\frac{n+k-c}{2} + d_H(C) + \left( \left\lceil \frac{n+k-c}{2r}\right\rceil - 1 \right) (\delta - 1) \leq n+1.
	\end{equation}
\end{cor}

\begin{proof}
	This follows straightforwardly from Corollaries \ref{coro1} and \ref{coro2}, the classical Singleton-like bound (see \Cref{Singleform}) on $C$,
	\[ \dim C + d_H(C) + \left( \left\lceil \frac{\dim C}{r}\right\rceil - 1 \right) (\delta - 1) \leq n+1, \]
	and the fact that $\dim(C)=\frac{n+k-c}{2}$.
\end{proof}

\begin{rem}
We say that an entanglement-assisted quantum $(r,\delta)$-LRC $Q(C)$ as above is \textit{pure} whenever its minimum distance equals $d_H(C)$. The Singleton-like inequality in \Cref{eq:EASing} depends on $d_H(C)$, and, analogously to \cite[Remark 34]{qlrc24}, for non-pure codes, the bound could be exceeded when replacing $d_H(C)$ with the actual minimum distance of $Q(C)$. The pure case allows us to state the following definition.
\end{rem}

\begin{defi}
	A pure EA quantum $(r,\delta)$-LRC that attains the bound in \Cref{eq:EASing} is said to be \textit{optimal}.
\end{defi}

\begin{rem}\label{rem:breakSingleton}
A family of impure (not entanglement-assisted) quantum $(r,\delta)$-LRCs exceeding the quantum Singleton-like bound \eqref{eq:SingletonQ}, or \eqref{eq:EASing} with $c=0$, was given in \cite{impureqlrcs}. Some parameters of codes in this family are $[[9,1,4;0]]_3$ and $[[15,1,6;0]]_5$, both with $(r,\delta)=(2,2)$, and $[[64,4,16;0]]_8$ with $(r,\delta)=(5,4)$. This family was obtained from the CSS construction.
\end{rem}

%\color{blue}
\begin{rem}\label{rem:hopelessAGcodes}
Suppose that  $C$ in \Cref{coro3} comes as an algebraic geometry code, that the local distance $\delta$ is upper bounded by the Goppa bound on the minimum Hamming distance of $C$ and that $C \neq C\cap C^{\perp_e}$. Then, the Goppa bound on $C\cap C^{\perp_e}$ is strictly greater than $\delta$ and Assumption \ref{assumption2} in \Cref{coro3}
holds. This means that a straightforward use of algebraic geometry codes cannot provide a quantum code exceeding \Cref{coro3} if its quantum distance is evaluated by $d_H(C)$. In our previous paper \cite{impureqlrcs}, we were able to construct quantum codes whose quantum distances exceed \Cref{coro3} by using the decreasing monomial-Cartesian codes \cite{camps21,Geil13} and the Feng-Rao bound on the quantum distance $d_H(C \setminus (C\cap C^{\perp_e}))$.
\end{rem}
%\color{black}

In the following sections, we present examples of optimal pure entanglement-assisted quantum $(r,\delta)$-LRCs coming from different families of classical codes: bivariate $J$-affine variety codes, BCH codes, and homothetic-BCH codes.

\section{EA $(r,\delta)$-QLRCs from bivariate $J$-affine-variety codes}

\label{sec:Jaf}

$J$-affine variety codes were introduced in \cite{QINP2}, and they are monomial-Cartesian codes as indicated in \cite{GFMC}. These codes evaluate multivariate polynomials, but in this section, we consider only the bivariate case.

Let us recall the definition in our setting. As above, $q$ is a power of a prime number $p$, and we will use the Euclidean inner product. Consider the polynomial ring $\mathbb{F}_q[X_1,X_2]$ and let $J \subseteq \{1,2\}$ be a set of indices indicating that the variable $X_j$, $j\in J$, is not evaluated at $0\in \mathbb{F}_q$.

Pick two integers $n_1$ and $n_2$ such that $n_j$ divides $q-1$ for $j\in J$, and $n_j-1$ divides $q-1$ otherwise. Let $\mathcal{I}$ be the ideal of $\mathbb{F}_q[X_1,X_2]$ generated by the binomials $X_j^{n_j}- 1$ for $j\in J$ and $X_j^{n_j}- X_j$ otherwise. Set
$$P = \{\boldsymbol{\alpha}_1,\dots,\boldsymbol{\alpha}_n\} \subseteq \mathbb{F}_q^2$$
to be the zero-set of $\mathcal{I}$, $n:=n_1n_2$. Let
$$E=\{0,1,\dots,n_1-1\}\times\{0,1,\dots,n_2-1\}$$
be the set of pairs corresponding to exponents of the monomials whose classes are evaluated by the linear evaluation map
$$
\mathrm{ev}_P: \mathcal{R}:= \frac{\mathbb{F}_q[X_1,X_2]}{\mathcal{I}} \rightarrow \mathbb{F}_q^{n}, \quad \mathrm{ev}_P(f)=\left(f(\boldsymbol{\alpha}_1),\dots,f(\boldsymbol{\alpha}_n)\right).
$$
Thus, for each set $\emptyset \neq \Delta \subseteq E$, the \textit{$J$-affine variety code} is the following $\mathbb{F}_q$-linear space:
$$C_\Delta^{P,J}:=\mathrm{span} \left\{\mathrm{ev}_P(X_1^{e_1}X_2^{e_2}) \;: \; (e_1,\,e_2)\in \Delta \right\} \subseteq \mathbb{F}_q^n.$$

\begin{pro}\label{rect}
	Keep the above notation. Set $J=\{1\}$. Let $n_1$ and $n_2$ be two integers such that $n_1, n_2-1 \mid q-1$ and $p \mid n_2$. Let $a, b$ be non-negative integers such that $$2a+2b \leq n_1-1 \; \mbox{ and } \; 2a +b \leq n_1 -2.$$ Consider the following set depicted in Figure \ref{fig:rect} (a):
	\begin{align*}
		\Delta_{a,b}:=&\left\{(e_1,e_2) \in E : 0\leq e_1 \leq a+b,\, 0\leq e_2\leq n_2-1\right\} \cup \\
		& \left\{(e_1,e_2) \in E : n_1-a \leq e_1 \leq n_1-1,\, 0\leq e_2\leq n_2-1\right\}
		\subseteq E.
	\end{align*}
	Then, the code $C_{\Delta_{a,b}}^{P,J}$ is an optimal $(r,\delta)=(2a+b+1,n_1-2a-b)$-LRC and gives rise to an optimal pure entanglement-assisted quantum $(r,\delta)$-LRC, $Q\left(C_{\Delta_{a,b}}^{P,J}\right)$, with parameters
	$$[[n_1n_2, (2a+1)n_2, n_1-(2a+b); (n_1-2(a+b)-1)n_2]]_q$$
	and locality $(r,\delta)=(2a+b+1,n_1-2a-b)$.

	Moreover, there exists an optimal pure EA $(r,\delta)$-QLRC with parameters as above, but with the roles of $n_1$ and $n_2$ interchanged.
\end{pro}
\begin{proof}
Let us prove the first statement. The second statement holds by symmetry after considering the following set (depicted in Figure \ref{fig:rect} (b)):
	\begin{align*}
		\Delta_{a,b}^\sigma:=&\left\{(e_1,e_2) \in E : 0\leq e_1\leq n_1-1 ,\, 0\leq e_2 \leq a+b \right\} \cup \\
		& \left\{(e_1,e_2) \in E : 0\leq e_1\leq n_1-1 ,\, n_2-a \leq e_2 \leq n_2-1\right\}
		\subseteq E
	\end{align*}
and replacing $(\{1\},n_1,n_2)$ in the above statement with $(\{2\},n_2,n_1)$.

We start by recalling that the length of $C_{\Delta_{a,b}}^{P,J}$ equals $\mathrm{card}(E)=n_1n_2$ and its dimension is equal to the cardinality of $\Delta_{a,b}$, i.e., $(2a+b+1)n_2$. The minimum distance of $C_{\Delta_{a,b}}^{P,J}$ is equal to the minimum distance of the code $C_{\Delta^1_{2a+b,n_2-1}}^{P,J}$, where the set $\Delta^1_{2a+b,n_2-1}$ is that defined in \cite[Proposition 4.1]{GFMC}, because they are isometric codes. Let us compute the minimum distance. As in the proof of \cite[Lemma 5.5]{GFMC}, $0 \in \mathbb{F}_q$ is not selected as an evaluation point for $X_1$, and codewords in $C_{\Delta^1_{2a+b,n_2-1}}^{P,J}$ are of the form
		$$\mathrm{ev}_P(X_1^a f)=\mathrm{ev}_P(X_1^a) * \mathrm{ev}_P(f),$$
		where $*$ denotes the Schur (or star/componentwise) product in $\mathbb{F}_q^n$, and $$f\in \left\{g \in \mathcal{R} : \mathrm{supp}(g) \subseteq \Delta_{a,b}\right\} \cup \left\{0\right\},$$ with $\mathrm{supp}(g)$ denoting the set of exponent vectors of the monomials of $g$. Then, the minimum distance of $C_{\Delta^1_{2a+b,n_2-1}}^{P,J}$, and therefore that of $C_{\Delta_{a,b}}^{P,J}$, equals $n_1-(2a+b)$ by \cite[Corollary 3.7 and Remark 3.8]{GFMC}.
		
The fact that $C_{\Delta_{a,b}}^{P,J}$ is an optimal $(r,\delta)=(2a+b+1,n_1-2a-b)$-LRC follows from \cite[Proposition 3.10]{GFMC}. Note that the MDS condition required to obtain the exact values of $(r,\delta)$ (the second statement of \cite[Proposition 3.10]{GFMC}) follows by arguing as above and applying the proof of \cite[Lemma 5.5]{GFMC}.
		
Now, \cite[Proposition 2.2]{QINP2} shows that, under the conditions in the statement, $\left(C_{\Delta_{a,b}}^{P,J}\right)^{\perp_e}=C_{\Delta^\perp}^{P,J}$, where
$$\Delta^\perp=\left\{(e_1,e_2) \in E : a+1\leq e_1\leq n_1-a-b-1,\, 0 \leq e_2 \leq n_2-1\right\}.$$
The intersection $C_{\Delta_{a,b}}^{P,J} \cap \left(C_{\Delta_{a,b}}^{P,J}\right)^{\perp_e}=C_{\Delta_\cap}^{P,J}$ is given by the set
$$\Delta_\cap=\left\{(e_1,e_2) \in E : a+1\leq e_1\leq a+b,\, 0 \leq e_2 \leq n_2-1\right\}.$$
Thus, $\dim \left(C_{\Delta_{a,b}}^{P,J} \right)^{\perp_e} = \left(n_1-2a-b-1\right)n_2$ and $\dim \left( C_{\Delta_{a,b}}^{P,J} \cap \left(C_{\Delta_{a,b}}^{P,J}\right)^{\perp_e}\right) =b n_2$.
		
Reasoning as in the second paragraph of this proof, the minimum distance of $C_{\Delta_\cap}^{P,J}$ is equal to the minimum distance of $C_{\Delta^1_{b-1,n_2-1}}^{P,J}$, which equals $n_1-b+1$. When $b \neq 0$ $$d\left(C_{\Delta_{a,b}}^{P,J}\right)=n_1-(2a+b) < n_1-b+1=d \left(C_{\Delta_{a,b}}^{P,J} \cap \left(C_{\Delta_{a,b}}^{P,J}\right)^{\perp_e}\right),$$ which proves that the derived quantum code $Q\left(C_{\Delta_{a,b}}^{P,J}\right)$ is pure. Note that purity remains valid when $b=0$, because then $C_{\Delta_\cap}^{P,J} = \{\mathbf{0}\}$ and $d(\{\mathbf{0}\}) = \infty$.
		
Finally, the parameters of the quantum code $Q\left(C_{\Delta_{a,b}}^{P,J}\right)$ can be computed from the expression given before the statement of \Cref{coro1}. Since Item (1) in \Cref{coro3} holds and $C_{\Delta_{a,b}}^{P,J}$ is an optimal $(r,\delta)$-LRC, one obtains that $Q\left(C_{\Delta_{a,b}}^{P,J}\right)$ is an optimal pure entanglement-assisted quantum $(r,\delta)$-LRC.
\end{proof}

\begin{figure}[h]
	\centering
	\begin{subfigure}[b]{0.45\textwidth}
		\centering
		\begin{tikzpicture}[y=0.7cm, x=0.7cm,font=\normalsize]
			
			\filldraw[fill=gray!30] (0,0) rectangle (2,6);
			\filldraw[fill=gray!30] (4,0) rectangle (6,6);
			\draw (2,0) -- (6,0);
			
			\node [below] at (0,0) {\scriptsize$0$};
			\node [below] at (1,0) {$\dots$};
			\node [below] at (2,0) {\scriptsize$a+b$};
			\node [below] at (3,0) {$\dots$};
			\node [below] at (4,0) {\scriptsize$n_1-a$};
			\node [below] at (5,0) {\scriptsize$\dots$};
			\node [below] at (6,0) {\scriptsize$n_1-1$};
			\node [left] at (0,0) {\scriptsize$0$};
			\node [left] at (0,3) {\scriptsize$\vdots$};
			\node [left] at (0,6) {\scriptsize$n_2-1$};
			
		\end{tikzpicture}
		\caption{The set $\Delta_{a,b}$}
	\end{subfigure}
	\hfill
	\begin{subfigure}[b]{0.45\textwidth}
		\centering
		\begin{tikzpicture}[y=0.7cm, x=0.7cm,font=\normalsize]
			
			\filldraw[fill=gray!30] (0,0) rectangle (6,2);
			\filldraw[fill=gray!30] (0,4) rectangle (6,6);
			\draw (0,2) -- (0,4);
			
			\node [below] at (0,0) {\scriptsize$0$};
			\node [below] at (3,0) {$\dots$};
			\node [below] at (6,0) {\scriptsize$n_1-1$};
			\node [left] at (0,0) {\scriptsize$0$};
			\node [left] at (0,1) {$\vdots$};
			\node [left] at (0,2) {\scriptsize$a+b$};
			\node [left] at (0,3) {$\vdots$};
			\node [left] at (0,4) {\scriptsize$n_2-a$};
			\node [left] at (0,5) {$\vdots$};
			\node [left] at (0,6) {\scriptsize$n_2-1$};
		\end{tikzpicture}
		\caption{The set $\Delta_{a,b}^\sigma$}
	\end{subfigure}
	\caption{Sets $\Delta_{a,b}$ and $\Delta_{a,b}^\sigma$ in Proposition \ref{rect}}
	\label{fig:rect}
\end{figure}
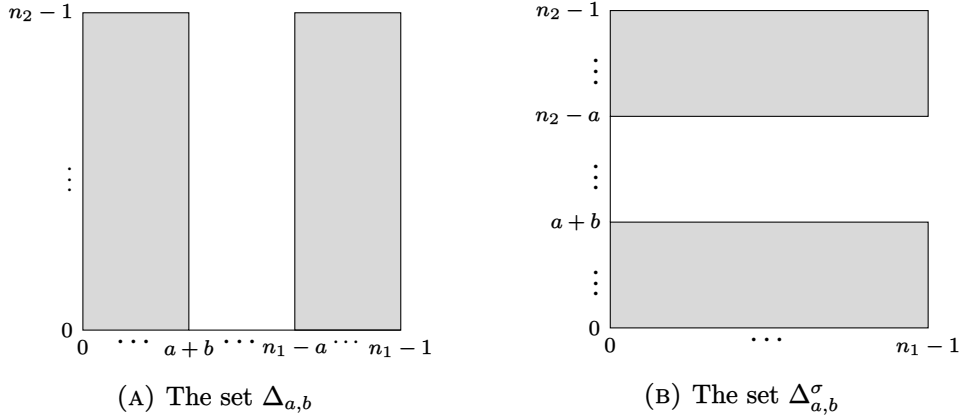

\begin{exa}
	Consider $J=\{1\}$ and the values $q=7$, $n_1=6$, $n_2=7$, and $a=b=1$. Let $\Delta_{1,1}$ be as defined in Proposition \ref{rect}. Since $2a+b=3 \leq 4=n_1-2$ and $2a+2b=4 \leq 5=n_1-1$, Proposition \ref{rect} shows that $Q\left(C_{\Delta_{a,b}}^{P,J}\right)$ is an optimal quantum $(r, \delta)= (4,3)$-LRC with parameters $[[42,21,3;7]]_7$, and therefore its parameters and locality reach the bound in \eqref{eq:EASing}.
\end{exa}

\section{EA $(r,\delta)$-QLRCs from BCH and homothetic-BCH codes}\label{sec:BCHhomo}

In this section, we consider BCH and homothetic-BCH codes as evaluation codes, as has been previously done in \cite{hbch25,bchqlrc}. Let us introduce them.

Recall that $q$ is a prime power and let $s \geq 2$ be an integer. BCH codes will be viewed as $q$-ary (or $q^2$-ary in the case $s=2\varsigma$, where $\varsigma$ is a positive integer) subfield-subcodes of univariate $q^s$-ary $\{1\}$-affine variety codes.

Let $N$ be an integer that divides $q^s-1$, and let $\mathcal{I}$ be the ideal of the polynomial ring $\F_{q^s}[X]$ generated by $X^N - 1$. Denote by $\beta \in \F_{q^s}$ a primitive $N$-th root of unity and let $U(N)=\{1,\beta,\dots,\beta^{N-1}\}$ be the zero set of $\mathcal{I}$. Consider the set $\Z_N:=\{0,1,\dots,N-1\}$ of representatives of the ring $\Z/N\Z$ of congruences modulo $N$, corresponding to exponents of the monomials whose classes are evaluated using the linear evaluation map
$$
\ev_{U(N)}:  \frac{\F_{q^s}[X]}{\mathcal{I}} \rightarrow \F_{q^s}^N \textrm{, } \quad \ev_{U(N)}(f)=\left(f(1), f(\beta),\dots,f(\beta^{N-1})\right).
$$
Thus, for each set $ \emptyset \neq \Delta \subseteq \Z_N$, the \textit{(univariate) $\{1\}$-affine variety code} of length $N$ given by $\Delta$ is the following $\F_{q^s}$-linear space:
$$C_\Delta^N:=\spn \left\{\ev_{U(N)}(X^e) \;: \; e\in \Delta \right\} \subseteq \F_{q^s}^N.$$

Let $a\neq 0$ and $b\geq 2$ be relatively prime integers, and let $S$ be a subset of $\Z_b$. The set $S$ is said to be \textit{$a$-complete modulo $b$} if $ac \pmod b \in S$ for all $c \in S$. Sets that are $a$-complete modulo $b$ are unions of $a$-cyclotomic cosets modulo $b$. BCH codes can be defined as follows (see \cite{Bier,Cas}).

\begin{defi}\label{def:BCH}
	Let $\Delta \subseteq \Z_N$ be a $q$-complete modulo $N$ set. The subfield-subcode $C_\Delta^N \cap \F_q^N : = \BCH_q^N(\Delta)$ is called the $q$-ary \textit{BCH code} of length $N$ given by $\Delta$.
\end{defi}

\sloppy Let $t$ denote the maximum number of consecutive elements in $\Delta$. Then, by \cite{Bier,Cas}, the code $\BCH_q^N(\Delta)$ has parameters $[N,\card \Delta]_q$ and its Euclidean dual $\left(\BCH_q^N(\Delta)\right)^{\perp_e}$ satisfies $\dis \left( \left(\BCH_q^N(\Delta)\right)^{\perp_e} \right) \geq t+1$.

\begin{rem}\label{rem:2s}
	Notice that when $s=2\varsigma$, $q^2$-ary BCH codes are defined simply by setting $q^2$ instead of $q$ in the above definition. The minimum distance bound holds by considering the Hermitian dual instead of the Euclidean dual in the above paragraph.
\end{rem}

Homothetic-BCH codes were introduced in \cite{hbch25}. They are obtained as an enlargement of BCH codes, resulting in codes with lengths that cannot be obtained by any standard BCH code as above. Let us introduce them.

Consider positive integers $n$ and $\lambda$ such that $n$ divides $N$, $\lambda < \frac{N}{n}$, and $\lambda n$ does not divide $N$. Recall that $\beta \in \F_{q^s}$ denotes a primitive $N$-th root of unity and set $\zeta_n:= \beta^{\frac{N}{n}}$. Let $P$ be the set
\begin{multline*}
P:= \{1,\zeta_n,\dots,\zeta_n^{n-1}, \beta,\beta\zeta_n,\dots,\beta\zeta_n^{n-1}, \dots, \beta^{\lambda-1},\beta^{\lambda-1}\zeta_n,\dots,\beta^{\lambda-1}\zeta_n^{n-1}\}= \\
\{\alpha_1,\dots,\alpha_{\lambda n}\},
\end{multline*}
whose vanishing ideal (in $\F_{q^s}[X]$) is
$$\mathcal{J}= \spn \{(X^n-1)(X^n-\beta^n) \cdots (X^n-\beta^{(\lambda-1)n}) \} \subseteq \F_{q^s}[X].$$
As before, the linear evaluation map
\begin{equation}\label{homoevamap}
\ev_P:  \frac{\F_{q^s}[X]}{\mathcal{J}} \rightarrow \F_{q^s}^{\lambda n} \textrm{, } \quad \ev_P(f)=\left(f(\alpha_1),\dots,f(\alpha_{\lambda n})\right)
\end{equation}
allows us to define homothetic-BCH codes as the following class of evaluation codes.

\begin{defi}
	Keep the above notation and let $\Delta \subseteq \Z_N$. The \textit{homothetic evaluation code} given by $\Delta$ is the $q^s$-ary code of length $\lambda n$:
	$$\mathcal{H}_\Delta^P:= \spn \{\ev_P(X^e) : e\in\Delta\} \subseteq \F_{q^s}^{\lambda n}.$$
	The $q$-ary \textit{homothetic-BCH (H-BCH) code} given by $\Delta$ is its subfield-subcode:
	$$\mathcal{S}_\Delta^{P,q}:= \mathcal{H}_\Delta^P \cap \F_q^{\lambda n}.$$
\end{defi}

Let $\Delta \subseteq \Z_N$ be a $q$-complete modulo $N$ set and let $t$ denote the maximum number of consecutive elements in $\Delta$. Then, by \cite{hbch25}, the H-BCH code $\mathcal{S}_\Delta^{P,q}$ has parameters $[\lambda n, \leq \card \Delta]_q$ and its Euclidean dual $\left(\mathcal{S}_\Delta^{P,q}\right)^{\perp_e}$ satisfies $\dis \left( \left(\mathcal{S}_\Delta^{P,q}\right)^{\perp_e} \right) \geq t+1$.

\begin{rem}
	Similar to \Cref{rem:2s}, notice that when $s=2\varsigma$, one can also introduce $q^2$-ary H-BCH codes by setting $q^2$ instead of $q$ in the above definition. A bound as in the previous paragraph also holds by setting the Hermitian dual instead of the Euclidean dual.
\end{rem}

\subsection{EA $(r,\delta)$-QLRCs from BCH codes}\label{subsec:BCH}

Keep the notation as above and consider the set
$$\mathcal{A}:=\{0,n,2n,\dots,N-n\}\subset \Z_N.$$
The following proposition recalls some facts regarding this set. They were proved in \cite[Propositions 18 and 20]{bchqlrc}.

\begin{pro}\label{prop1}
	Let $n$ and $N$ be as above and let $a$ be a positive integer relatively prime to $N$. Then, the following statements hold.
	\begin{enumerate}
		\item The set $\mathcal{A}$ is $a$-complete modulo $N$.
		\item The set $\mathcal{A}+B$ is $a$-complete modulo $N$ for any $a$-complete modulo $n$ set $B$.
		\item $\dis\left(C_{\mathcal{A}}^N\right)=n$.
	\end{enumerate}
\end{pro}

The next proposition was also proved in \cite{bchqlrc}, originating from \cite[Theorem 12]{QZF2021}. It allows us to regard (Euclidean or Hermitian) duals of BCH codes as $(r,\delta)$-LRCs.

\begin{pro}\label{prop:qui}
   Let $A \subseteq \Z_N$ be a $q$-complete modulo $N$ set and let $B\subseteq \Z_n$ be a $q$-complete modulo $n$ set. Let $d_A$ and $d_B^{\perp_e}$ denote, respectively, the minimum distances of the codes $C_A^N$ and $\left(C_B^N\right)^{\perp_e}$. Consider a positive integer $\delta$ such that $\delta \leq \min \left\{d_A, d_B^{\perp_e}\right\}$. Then, the code $\left(\BCH_q^N(A+B)\right)^{\perp_e}$ is a $q$-ary $(r,\delta)$-LRC with locality $\left(d_A-\delta+1,\delta\right)$.

   Additionally, in the case $s=2\varsigma$, the above statement also holds setting $\perp_h$ instead of $\perp_e$ and $q^2$ instead of $q$.
\end{pro}

In the rest of the paper, for a set $B\subseteq \{0,1,\dots,n-1\}$ and an integer $\ell$, we define the set
$$\ell B := \left\{i \in \{0,1,\dots,n-1\} : i \equiv \ell b \pmod n \textrm{ for some } b\in B\right\}.$$
In particular, we denote the set $(-\ell)B$ as $-\ell B$ in what follows.

\begin{lem}\label{lem:Bdual}
	Keep the notation from the beginning of Subsection \ref{subsec:BCH}. Let $B \subseteq \Z_n$ be a $q$-complete modulo $n$ set. Then, the following statements hold.
	\begin{enumerate}
	\item The Euclidean dual $\left(\BCH_q^N(\mathcal{A}+B)\right)^{\perp_e}$ is the code
	$$\left(\BCH_q^N(\mathcal{A}+B)\right)^{\perp_e}=\BCH_q^N\left(\mathcal{A}+ \left(\{0, \ldots, n-1\}\setminus (-B)\right)\right).$$
	\item In the case $s=2\varsigma$, the Hermitian dual $\left(\BCH_{q^2}^N(\mathcal{A}+B)\right)^{\perp_h}$ is the code
	$$\left(\BCH_{q^2}^N(\mathcal{A}+B)\right)^{\perp_h} = \BCH_{q^2}^N\left(\mathcal{A}+ \left(\{0, \ldots, n-1\}\setminus(-qB)\right)\right).$$
	\end{enumerate}
\end{lem}

\begin{proof}
Let us prove the first claim. Let $\Delta\subseteq \{0,\dots,N-1\}$ be a $q$-complete modulo $N$ set. We denote $N-\Delta=\left\{N-e \pmod N : e\in \Delta\right\}$. Then, we have $\left(\BCH_q^N(\Delta)\right)^{\perp_e}=\BCH_q^N\left(\{0,\dots,N-1\} \setminus (N-\Delta)\right)$.

It is also easy to see that any $i \in \{0,\dots,N-1\}$ can be written as $i=a+r$, with $a\in\mathcal{A}$, $r\in\{0,\dots,n-1\}$ and that $\{0,\dots,N-1\} = (\mathcal{A}+(-B)) \cup \left(\mathcal{A}+\left(\{0,\dots,n-1\} \setminus (-B)\right)\right)$. Then, $N-\mathcal{A} = \mathcal{A}$, $N-(\mathcal{A}+B) = \mathcal{A} + (-B)$,
and
$$\{0, \dots, N-1\} \setminus (\mathcal{A}+(-B)) = \mathcal{A} + \left(\{0, \dots, n-1\}\setminus (-B)\right),$$
showing the first claim. The second one can be shown in an analogous way.
\end{proof}

Our goal is to construct EA $(r,\delta)$-qLRCs $Q'(C)$ from certain BCH codes $C$ of the form $\BCH_q^N(\mathcal{A}+B)$ (or $\BCH_{q^2}^N(\mathcal{A}+B)$), whose Euclidean (or Hermitian) duals are classical $(r,\delta)$-LRCs. To determine the amount of entanglement, we need the following proposition, which immediately follows from \Cref{lem:Bdual}.

\begin{pro}\label{prop:hulldim}
	Keep the above notation and let $B \subseteq \Z_n$ be a $q$-complete modulo $n$ set. Then, the following statements hold.
	\begin{enumerate}
		\item The Euclidean amount of entanglement
		$$c=\dim \left(\BCH_q^N(\mathcal{A}+B)\right) - \dim \left(\BCH_q^N(\mathcal{A}+B) \cap \left(\BCH_q^N(\mathcal{A}+B)\right)^{\perp_e}\right)$$
		is equal to $\frac{N}{n} \card(B \cap -B)$.
		\item In the case $s=2\varsigma$, the Hermitian amount of entanglement
		$$c=\dim \left(\BCH_{q^2}^N(\mathcal{A}+B)\right) - \dim \left(\BCH_{q^2}^N(\mathcal{A}+B) \cap \left(\BCH_{q^2}^N(\mathcal{A}+B)\right)^{\perp_h}\right)$$
		is equal to $\frac{N}{n} \card(B \cap (-qB))$.
	\end{enumerate}
\end{pro}

The following proposition will be useful to show the purity of our codes. The next lemma will be used in the proof.

\begin{lem}\label{lem:purity}
   Let $C \subseteq \F_q^n$ be an $[n,k, n-k+1]_q$ MDS code
   and let $D \subseteq C$ be any linear subcode.
   Then, we have either $\dis\left(C \setminus D\right) = \dis(C)$ or $D=C$.
\end{lem}
\begin{proof}
   Assume $\dis\left(C \setminus D\right) \neq \dis(C)$, that is, $D$ contains all the minimum weight codewords in $C$,
   whose Hamming weight is $n-k+1$. Let $G = \left(I_{k\times k} \mid A_{k \times (n-k)}\right)$ be a systematic generator matrix of $C$. The Hamming weight of the codeword in every row of $G$ is exactly $n-k+1$, which shows that the rows of $G$ span $D$ and $D=C$, proving the statement.
\end{proof}

\begin{pro}\label{pro:purity}
    Keep the notation as above. Let $B \subseteq \Z_n$ be a nonempty set of consecutive
    integers which is also $q$-complete (respectively, $q^2$-complete in the case $s=2\varsigma$) modulo $n$.
    Then, we have either
    \begin{align*}
      & \dis\left(\left( \BCH_q^N(\mathcal{A}+B) \right)^{\perp_e}\right)\\
      =& \dis\left( \left( \BCH_q^N(\mathcal{A}+B) \right)^{\perp_e} \Big\backslash \left( \left( \BCH_q^N(\mathcal{A}+B) \right)^{\perp_e} \cap \BCH_q^N(\mathcal{A}+B)  \right)\right)
      \end{align*}
    or
    $$\left( \BCH_q^N(\mathcal{A}+B) \right)^{\perp_e} = \left( \BCH_q^N(\mathcal{A}+B) \right)^{\perp_e} \cap \BCH_q^N(\mathcal{A}+B)$$ (respectively,
    \begin{align*}
      & \dis\left(\left( \BCH_{q^2}^N(\mathcal{A}+B) \right)^{\perp_h}\right) \\
      = & \dis\left( \left( \BCH_{q^2}^N(\mathcal{A}+B) \right)^{\perp_h} \Big\backslash \left( \left( \BCH_{q^2}^N(\mathcal{A}+B) \right)^{\perp_h} \cap \BCH_{q^2}^N(\mathcal{A}+B)  \right)\right)
      \end{align*}
    or
    $$\left( \BCH_{q^2}^N(\mathcal{A}+B) \right)^{\perp_h} = \left( \BCH_{q^2}^N(\mathcal{A}+B) \right)^{\perp_h} \cap \BCH_{q^2}^N(\mathcal{A}+B)$$
    in the case $s=2\varsigma$).
    %\color{black}
  \end{pro}

  \begin{proof}
    Assume that $B$ is $q$-complete modulo $n$ and consider duality with respect to the Euclidean inner product. Let
    $$B^{\perp_e} := \left\{0, \ldots, n-1\right\} \setminus (-B).$$
    Then, by \Cref{lem:Bdual} we have $\left(\BCH_q^N(\mathcal{A}+B)\right)^{\perp_e} = \BCH_q^N\left(\mathcal{A}+ B^{\perp_e}\right)$ and
    $$\left(\BCH_q^N(\mathcal{A}+B)\right)^{\perp_e} \cap \BCH_q^N(\mathcal{A}+B) = \BCH_q^N\left(\mathcal{A}+ \left(B^{\perp_e} \cap B\right)\right).$$

Recall that $\beta\in\F_{q^s}$ denotes a primitive $N$-th root of unity and $\zeta_n=\beta^{\frac{N}{n}}$ is a primitive $n$-th root of unity. Consider the set $U(n)=\left\{1,\zeta_n,\dots,\zeta_n^{n-1}\right\}$ (which is a subset of $U(N)=\left\{1,\beta,\dots,\beta^{N-1}\right\}$) and the $\{1\}$-affine variety code $C_B^n \subseteq \F_{q^s}^n$. Let $\BCH_q^n(B) \subseteq \F_q^n$ be its $q$-ary subfield-subcode, which is an $\left[n, \card(B), n-\card(B)+1\right]_q$ MDS code. Then, $\left(\BCH_q^n(B)\right)^{\perp_e}$ is an $\left[n, n-\card(B), \card(B)+1\right]_q$ MDS code, $\left(\BCH_q^n(B)\right)^{\perp_e}= \BCH_q^n\left(B^{\perp_e}\right)$ and $\left(\BCH_q^n(B)\right)^{\perp_e} \cap \BCH_q^n(B) = \BCH_q^n\left(B^{\perp_e} \cap B\right)$.
	
	By \Cref{lem:purity}, we have either
	\begin{equation}\label{proof:cond1}
		\dis\left(\left(\BCH_q^n(B)\right)^{\perp_e} \Big\backslash \left(\left(\BCH_q^n(B)\right)^{\perp_e} \cap \BCH_q^n(B)\right)\right) = \dis\left(\left(\BCH_q^n(B)\right)^{\perp_e}\right),
	\end{equation}
    or $\left(\BCH_q^n(B)\right)^{\perp_e} = \left(\BCH_q^n(B)\right)^{\perp_e} \cap \BCH_q^n(B)$, which is equivalent
    to
    $$\BCH_q^N(\mathcal{A}+B)^{\perp_e} = \BCH_q^N(\mathcal{A}+B)^{\perp_e} \cap \BCH_q^N(\mathcal{A}+B),$$
    thus satisfying the second equality of the statement.

   	Assume then that Equality (\ref{proof:cond1}) holds and let $\boldsymbol{b} \in \left(\BCH_q^n(B)\right)^{\perp_e} \Big\backslash \left(\left(\BCH_q^n(B)\right)^{\perp_e} \cap \BCH_q^n(B)\right)$ be a minimum weight ($=\card(B)+1$) codeword.
    Let $f(X) = \sum_{i \in B^{\perp_e}} a_i X^i \in \F_{q^s}[X]$ be a polynomial such that $\boldsymbol{b}=\ev_{U(n)}(f)$. Define
    \begin{equation}\label{proof:defg}
    	g(X) = \frac{X^N-1}{X^n-1} = \sum_{i=0}^{\frac{N}{n}-1} X^{in} \in \F_{q^s}[X].
    \end{equation}
    The set of roots of $g(X)$ is
    $U(N) \setminus U(n)$. Moreover, $f(X)$ and $g(X)$ are relatively prime, and $f(X)g(X)$ belongs to
    $\spn\left\{ X^i : i \in \mathcal{A}+B^{\perp_e} \right\} \subseteq \F_{q^s}[X]$. The above observations imply that the codeword $\ev_{U(N)}(fg)$ has Hamming weight $\card(B)+1$ and belongs to $\left(C_{\mathcal{A}+B}^N\right)^{\perp_e} \Big\backslash \left(\left(C_{\mathcal{A}+B}^N\right)^{\perp_e} \cap C_{\mathcal{A}+B}^N\right)$.

    Notice that, if $a\in U(n)$, we know that $f(a)\in\F_q$, and the right-hand side of the second equality in (\ref{proof:defg}) shows that $g(a)=\frac{N}{n}$, which belongs to the prime field $\F_p \subseteq \F_q$. Otherwise, if $a\in U(N) \setminus U(n)$, the left-hand side of the second equality in (\ref{proof:defg}) shows that $g(a)=0 \in \F_q$, as the numerator vanishes and the denominator does not. Therefore, $f(a)g(a) \in \F_q$ for any $a\in U(N)$, and $\ev_{U(N)}(fg)$ is a codeword of Hamming weight $\card(B)+1$ belonging to $ \left( \BCH_q^N(\mathcal{A}+B) \right)^{\perp_e} \Big\backslash \left( \left( \BCH_q^N(\mathcal{A}+B) \right)^{\perp_e} \cap \BCH_q^N(\mathcal{A}+B)  \right)$. This completes the proof for the Euclidean inner product.

    The claim for the Hermitian inner product follows by repeating
    the above argument with $B^{\perp_h} = \left\{0, \dots, n-1\right\} \setminus (-qB)$.
  \end{proof}

Let us define two sets that will be used in place of $B$ in the BCH codes $\BCH_q^N(\mathcal{A}+B)$ (or $\BCH_{q^2}^N(\mathcal{A}+B)$) that we construct to provide $(r,\delta)$-QLRCs. Such sets will satisfy the assumption on $B$ in \Cref{pro:purity}, giving rise to pure quantum codes
with or without entanglement assistance. For two integers $u, v$ such that $1 \leq u \leq v \leq n-1$, we define
$$B_{u,v}:=\left\{u,\dots,v\right\} \subseteq \Z_n.$$
For a single integer $1\leq u \leq \frac{n}{2}$, we define
$$B_u:=\left\{e : \frac{n}{2}-u < e < \frac{n}{2}+u\right\} \subseteq \Z_n.$$

The sets $B_{u,v}$ and $B_u$ are considered in \cite{bchqlrc} to obtain pure optimal $(r,\delta)$-QLRCs. In that context, entanglement assistance is not considered, and our results on these sets are deeper and different since we get the exact parameters of pure optimal EA $(r,\delta)$-QLRCs.

Let us begin with the Euclidean constructions.

\begin{thm}\label{thm:qminusone}
	Keep the notation from the beginning of Subsection \ref{subsec:BCH}. Assume that $2 \leq n \mid q-1$. Let $u, v$, and $B_{u,v}$ be as above. Then, the code $\left(\BCH_q^N(\mathcal{A}+B_{u,v})\right)^{\perp_e}$ is a classical $\left[N,N-\frac{N}{n}(v-u+1), v-u+2\right]_q$ optimal LRC with locality $(r,\delta)=\left(n-v+u-1,v-u+2\right)$. Therefore, the EAQECC $Q'\left(\BCH_q^N(\mathcal{A}+B_{u,v})\right)$ is an optimal pure EA $(r,\delta)$-QLRC with parameters
	$$\left[\left[N,N-2\frac{N}{n}(v-u+1)+c, v-u+2; c\right]\right]_q,$$
	where $c = \frac{N}{n}\max\left\{0, 1 + \min\left\{n-2u, 2v - n\right\}\right\}$, and locality $$(r,\delta)=\left(n-v+u-1,v-u+2\right).$$
\end{thm}

\begin{proof}
	Notice that $B_{u,v}$ is a $q$-complete modulo $n$ set as $q \equiv 1 \pmod{n}$ and $\card(B_{u,v})=v-u+1$.
	By Propositions \ref{prop1} and \ref{prop:qui} and the BCH bound, the code
	$\left(\BCH_q^N(\mathcal{A}+B_{u,v})\right)^{\perp_e}$ is a classical $\left[N, N - \frac{N}{n}\card(B_{u,v}), \card(B_{u,v})+1\right]_q$ LRC
	with locality $\left(n-\card(B_{u,v}), \card(B_{u,v})+1\right)$. These parameters and locality attain the Singleton bound (\ref{Singleform}) because the following sequence of equalities holds:
	\begin{multline*}
		\left(N-\frac{N}{n}\card(B_{u,v})\right) + (\card(B_{u,v})+1) + \\ \left(\left\lceil \frac{N-\frac{N}{n}\card(B_{u,v})}{n-\card(B_{u,v})}\right\rceil-1\right)\card(B_{u,v}) =\\
		 \frac{N}{n}\left(n-\card(B_{u,v})\right) + (\card(B_{u,v})+1) + \left(\left\lceil \frac{N}{n}\right\rceil-1\right)\card(B_{u,v}) =\\
		N+1.
	\end{multline*}
	
	Hence, by \Cref{coro3}, the EAQECC $Q'\left(\BCH_q^N(\mathcal{A}+B_{u,v})\right)$ is an optimal pure (by \Cref{pro:purity}) EA $(r,\delta)$-QLRC with parameters and locality as in the statement.
	The Euclidean amount of entanglement $c$ is given by \Cref{prop:hulldim}
	and the facts that $-B_{u,v} = \left\{n-v,\dots, n-u\right\}$ and
	\begin{eqnarray*}
		\card\left( B_{u,v} \cap -B_{u,v} \right)& =&
		\max\left\{0, 1+ \min\left\{n-u,v\right\} - \max\left\{u, n-v\right\}\right\}\\
		&=& \max\left\{0, 1 + \min\left\{n-2u, 2v - n\right\}\right\}.
	\end{eqnarray*}
\end{proof}

\begin{thm}\label{thm:qplusone}
	Keep the notation from the beginning of Subsection \ref{subsec:BCH}. Assume that $2 \leq n \mid q+1$. Let $u$ and $B_u$ be as above. Then, the code $\left(\BCH_q^N(\mathcal{A}+B_u)\right)^{\perp_e}$ is a classical $\left[N,N-\frac{N}{n}(2u - ((n+1) \bmod 2)), 2u - ((n+1) \bmod 2) + 1\right]_q$ optimal LRC with locality $(r,\delta)=\left(n-2u + ((n+1) \bmod 2),2u - ((n+1) \bmod 2) + 1\right)$. Therefore, the EAQECC $Q'\left(\BCH_q^N(\mathcal{A}+B_u)\right)$ is an optimal pure EA $(r,\delta)$-QLRC with parameters
	$$\left[\left[N,N-\frac{N}{n}(2u - ((n+1) \bmod 2)), 2u - ((n+1) \bmod 2) + 1; c\right]\right]_q,$$
	where $c = \frac{N}{n}(2u - ((n+1) \bmod 2))$, and the locality is $$(r,\delta)=\left(n-2u + ((n+1) \bmod 2),2u - ((n+1) \bmod 2) + 1\right).$$
\end{thm}

\begin{proof}
	The proof follows similarly to that of \Cref{thm:qminusone}, writing $B_u$ instead of $B_{u,v}$ and taking into account that $B_u$ is a $-1$-complete modulo $n$ set because $n$ divides $q+1$. Then, $-B_u=B_u$ and $\card(B_u \cap -B_u) = \card(B_u) = 2u - ((n+1) \bmod 2)$.
\end{proof}

Hereafter, we will assume that $s=2\varsigma$ to provide the analogue of \Cref{thm:qplusone} by using the Hermitian inner product, that is, considering $\left(\BCH_{q^2}^N(\mathcal{A}+B_u)\right)^{\perp_h}$
instead of $\left(\BCH_q^N(\mathcal{A}+B_u)\right)^{\perp_e}$. For that purpose, the forthcoming Propositions \ref{prop:hdimo} and \ref{prop:hdime} provide the amount of entanglement $\card(B_u \cap - qB_u)$ given in \Cref{prop:hulldim}.

Our next result assumes that $q$, $n$, and $u$ are as follows: $q$ is odd, and $n$ is an even integer such that $2 < n=m^2 + 1$ (for some positive integer $m$) which divides $q^2 +1$. We also assume that $q \equiv \pm m \pmod n$. Let $\ell$ and $r$ be nonnegative integers with $r < m$ and $0 < \ell m + r \leq n/2$, and consider $u := \ell m + r$.

\begin{pro}\label{prop:hdimo}
Let $q, m, \ell, r$, and $u$ be as above. Let $B_u$ be as introduced above \Cref{thm:qminusone}. Then, the cardinality of the intersection $B_u \cap - qB_u$ equals
\begin{equation}
\card\left(B_u \cap - qB_u\right) =
\begin{cases}
 4\ell^2 + 4r - 3 & \text{if } r \leq \ell, \\
 (2\ell+1)^2 & \text{if } \ell < r \leq m - \ell, \\
 (2\ell+1)^2 + 4(r + \ell - m) & \text{if } r > m - \ell.
\end{cases}
\end{equation}
\end{pro}

\begin{proof}
Since $n$ is even, $n/2$ is an integer and then $B_u = \left\{ n/2 - u + 1, \dots, n/2 + u - 1 \right\}$. Setting $\mathfrak{u} = u-1$, elements in $B_u$ can be expressed as $n/2 + j$, $j \in \left[ -\mathfrak{u}, \mathfrak{u}\right] \cap \mathbb{Z}$. Elements in $- qB_u$ are of the form $ -q(n/2 + i) \mod n$, $i \in \left[ -\mathfrak{u}, \mathfrak{u}\right] \cap \mathbb{Z}$. \textit{The elements in $B_u \cap - qB_u$ correspond to pairs of integers $(i,j)$ that satisfy}
\begin{equation}
\label{CG1}
n/2 + j \equiv -q(n/2 + i) \pmod n.
\end{equation}
Because $q$ is odd (let $q = 2k+1$) and $n$ is even, we observe that
$$ -q(n/2) = -(2k+1)(n/2) = -kn - n/2 \equiv -n/2 \equiv n/2 \pmod n. $$
By (\ref{CG1}), \textit{our goal is to count the number of integer pairs $(i,j)$ in the set}
$$\left( [-\mathfrak{u}, \mathfrak{u}] \cap \mathbb{Z} \right)^2 := \left([-\mathfrak{u}, \mathfrak{u}] \cap \mathbb{Z}\right) \times \left([-\mathfrak{u}, \mathfrak{u}] \cap \mathbb{Z}\right)$$
\textit{that satisfy} $j \equiv -qi \pmod n$.

By assumption, $q \equiv \pm m \pmod n$; then if $q \equiv m \pmod n$, one gets that $j \equiv -mi \pmod n$ and $j \equiv mi \pmod n$ whenever $q \equiv -m \pmod n$. Because the square $\left( [-\mathfrak{u}, \mathfrak{u}] \cap \mathbb{Z} \right)^2$ is symmetric across the axes (i.e., closed under sign changes in the coordinates), we only need to count the integer pairs $(i,j)$ which are included in the linear lattice
$$ L = \left\{(i,j) \in \mathbb{Z}^2 \mid j \equiv mi \pmod n \right\}. $$

$L$ is orthogonal and generated by the basis vectors $v_1 = (1, m)$ and $v_2 = (-m, 1)$ because $m^2 \equiv -1 \pmod n$. Thus, elements in $L$ have the form $(a - bm, am + b)$ with $a, b \in \Z$. Since $u = \ell m + r$ and noting that the pairs $(i,j)$ that we have to count belong to the square $\left( [-\mathfrak{u}, \mathfrak{u}] \cap \Z \right)^2$, we get
\begin{align}
\left|a - bm\right| &\leq \ell m + r - 1, \label{eq:bound1} \\
\left|am + b\right| &\leq \ell m + r - 1. \label{eq:bound2}
\end{align}

Thus, \textit{determining $\card\left(B_u \cap - qB_u\right)$ is equivalent to counting the number of integer pairs $(a,b)$ satisfying Inequalities \eqref{eq:bound1} and \eqref{eq:bound2}}. Let us compute this number.

The points such that $\left|a\right|,\,\left|b\right| \leq \ell-1$ satisfy Inequalities \eqref{eq:bound1} and \eqref{eq:bound2}. This is because, by the triangle inequality property, we have that
$$\left|a-bm\right|,\,\left|am+b\right| \leq (\ell-1)m + (\ell-1) = \ell m - m + \ell - 1,$$
a value that is less than or equal to $\ell m + r - 1$ if and only if $\ell - m \leq r$, which is true as $r\geq 0$ and $\ell < m$, because $u=\ell m + r \leq \frac{n}{2}=\frac{m^2+1}{2}$.

Consider
$$G_0 := \left\{ (a, b) \in \Z^2 : \left|a\right| \le \ell, \left|b\right| \le \ell \right\}$$
as our initial candidate grid, which contains exactly $(2\ell + 1)^2$ points, and let
$$G \subset \Z^2$$
be \textit{the set of solutions $(a,b)$ corresponding to Inequalities \eqref{eq:bound1} and \eqref{eq:bound2}}. Define the map
$$R \colon \Z^2 \to \Z^2, \quad R(a, b) = (-b, a),$$
which can be thought of geometrically as a $90^\circ$ counterclockwise rotation. The grid $G_0$ is invariant under the above map, that is, $R(G_0) = G_0$, and the same holds for $G$.
Indeed, evaluating Inequality \eqref{eq:bound1} for a rotated point $R(a,b)$ gives Inequality \eqref{eq:bound2} as $\left|-b - am\right| = \left|am + b\right| \le \ell m + r - 1$, and vice versa ($\left|-bm + a\right| = \left|a - bm\right| \le \ell m + r - 1$).

Then, the number of valid points $(a,b) \in G$ fixing a vertical line $a=v_0$ is the same as the number of valid points fixing the horizontal line $b=v_0$. This allows us to restrict our search to the possible range of values of the single variable $a$.

Let us prove that $G$ contains no valid points for any vertical line $a \ge \ell + 2$. Recalling that $\mathfrak{u} = \ell m + r - 1$, Inequalities \eqref{eq:bound1} and \eqref{eq:bound2} provide the following bounds for $b$:
	\begin{align}
		\left|a - bm\right| \le \mathfrak{u} &\iff \frac{a - \mathfrak{u}}{m} \le b \le \frac{a + \mathfrak{u}}{m} \quad \mbox{and} \quad \label{eq:boundb1} \\
		\left|am + b\right| \le \mathfrak{u} &\iff -\mathfrak{u} - am \le b \le \mathfrak{u} - am. \label{eq:boundb2}
	\end{align}
In particular, for an integer solution $b$ to exist, it must hold that $\frac{a - \mathfrak{u}}{m} \leq \mathfrak{u} - am$. Multiplying by $m>0$, we obtain $a - \mathfrak{u} \le m\mathfrak{u} - am^2$, or equivalently
\begin{equation}\label{eq:generala}
	a(m^2 + 1) \le \mathfrak{u}(m + 1) = (\ell m + r - 1)(m+1).
\end{equation}
Then, the fact that $r \leq m-1$ and $a \geq \ell +2$ allows us to prove that
\begin{multline*}
\ell m^2 + 2m^2 + \ell + 2 = (\ell + 2)(m^2 + 1) \leq a(m^2 + 1) \leq \\(\ell m + m - 2)(m + 1) = \ell m^2 + m^2 + \ell m - m - 2,
\end{multline*}
which implies
$$ \ell m^2 + 2m^2 + \ell + 2 \le \ell m^2 + m^2 + \ell m - m - 2,$$
that is,
$$ m(m - \ell) + m + \ell + 4 \le 0,$$
a contradiction with the facts that $m,\,\ell\geq 0$ and $\ell < m$.

Now, we search for the points $(a,b)\in G_0$ on the vertical line $a=\ell$ that are valid. From Inequality \eqref{eq:bound2}, we have:
$$ -(\ell m + r - 1) \le \ell m + b \le \ell m + r - 1.$$
The lower bound is satisfied trivially for all $b \ge -\ell$, as $-\ell \geq -2\ell m - r +1$ (notice that for $\ell=0$ we always have $r\geq 1$ and for $\ell \geq 1$ one has $\ell(2m-1)\geq 1-r$). However, the upper bound restricts to $b \le r - 1$.

We perform the rest of the proof by analyzing the parameter $r$, according to the following cases that match the statement:

\textit{Case 1: $r \leq \ell$.} By the definition of $G_0$, which makes $b\leq \ell$, and the fact that $r \leq \ell$, we have that the points on the vertical line $a=\ell$ failing Inequality \eqref{eq:bound2} are those in the set $$E_1 = \left\{(\ell, b) \in \Z^2 : r \le b \le \ell\right\},$$ whose cardinality equals $\ell - r + 1$. Additionally, Inequality \eqref{eq:bound1} for all points on this line is
\begin{equation}\label{eq:l4.6}
-\ell m - r +1 \leq \ell - b m \leq \ell m + r - 1,
\end{equation}
which holds except for the bottom corner $(\ell, -\ell)$. As established previously, evaluating Inequality \eqref{eq:bound1} at any point is equivalent to evaluating Inequality \eqref{eq:bound2} at its rotated point. By the rotational invariance $R(G)=G$ and $R(G_0)=G_0$, the sets $E_2 = R(E_1)$, $E_3 = R^2(E_1)$, and $E_4 = R^3(E_1)$ (which contains the corner $(\ell, -\ell)$) contain the remaining points in $G_0$ that fail either Inequality \eqref{eq:bound1} or \eqref{eq:bound2}. Moreover, these four sets are mutually disjoint. Therefore, the candidate for the cardinality of the intersection $B_u \cap -qB_u$ is
\begin{multline*}
\card\left(B_u \cap -qB_u\right) = \card(G_0) - \sum_{k=1}^{4} \card(E_k) = \\ (2\ell + 1)^2 - 4(\ell - r + 1) = 4\ell^2 + 4r - 3.
\end{multline*}
This completes the proof for Case 1 after noticing that, in this case, no valid points on the vertical line $a=\ell+1$ exist. Indeed, from Equation \eqref{eq:generala}, we obtain
$$(\ell+1)(m^2+1) \leq (\ell m + r -1)(m+1).$$
However, the difference between the left and right sides is positive, since, on the one hand,
\begin{equation}\label{eq:l+1}
(\ell + 1)(m^2 + 1) - (\ell m + r - 1)(m + 1) = m^2 - m(\ell + r - 1) + \ell - r + 2,
\end{equation}
and, on the other hand, $\ell m + r = u \leq \frac{n}{2}=\frac{m^2+1}{2}$, which implies $m^2 \ge 2\ell m + 2r - 1$, and then the above difference is greater than or equal to:
$$(2\ell m + 2r - 1) - m\ell - mr + m + \ell - r + 2 = m(\ell - r + 1) + \ell + r + 1 >0.$$

\textit{Case 2: $\ell < r \leq m - \ell$.} The reasoning in the paragraph above Case 1 proves that Inequality \eqref{eq:bound2} is satisfied for all points $(\ell,b) \in G_0$, and Inequality \eqref{eq:bound1} also holds for these points; see Inequalities \eqref{eq:l4.6}. This proves that $G_0 \subseteq G$ in this case. Let us see that no valid points with $a = \ell + 1$ exist, leading to $G_0=G$. Indeed, factoring out $m$ from Equality \eqref{eq:l+1}, we obtain
$$ m(m - \ell - r + 1) + \ell - r + 2 \le 0. $$
However, using the Case 2 hypotheses $m - \ell - r \ge 0$ and that $m-r\geq \ell$, we have:
$$m(m - \ell - r + 1) + \ell - r + 2 \geq m + \ell - r + 2 \geq 2 \ell + 2 >0,$$
a contradiction. This completes the proof for Case 2, showing that
$$ \card\left(B_u \cap -qB_u\right) = (2\ell + 1)^2.$$

It only remains to consider the last case.

\textit{Case 3: $r > m - \ell$.} Notice that $\ell \neq 0$ in this case and the fact that $u\leq \frac{n}{2}$ implies $2\ell m \leq 2\ell m + 2r \leq m^2+1$, that is, $2\ell \leq m + \frac{1}{m}$. Then, since $m>2$ is odd, we have $m>2\ell$, and also $m-\ell > \ell$. As above, we have $G_0 \subseteq G$. Let us prove that, now, there exist valid points with $a = \ell + 1$. From Inequalities \eqref{eq:boundb1} and \eqref{eq:boundb2}, we have:
$$ \max\left( \frac{\ell + 1 - \mathfrak{u}}{m}, -\mathfrak{u} - (\ell + 1)m \right) \le b \le \min\left( \frac{\ell + 1 + \mathfrak{u}}{m}, \mathfrak{u} - (\ell + 1)m \right).$$
Let us compute an upper bound for $b$. Since $\mathfrak{u} = \ell m + r - 1$, then
$$ \frac{\ell + 1 + (\ell m + r - 1)}{m} = \frac{\ell m + \ell + r}{m} = \ell + \frac{\ell + r}{m} > 0,$$
and
$$ \mathfrak{u} - (\ell + 1)m = (\ell m + r - 1) - \ell m - m = r - m - 1 < 0.$$
This shows that the upper bound for $b$ is
$$ b \le r - m - 1.$$
With respect to a lower bound, we have
$$ -\mathfrak{u} - (\ell + 1)m = -(\ell m + r - 1) - \ell m - m = -2\ell m - m - r + 1 $$
and
$$ \frac{\ell + 1 - (\ell m + r - 1)}{m} = \frac{-\ell m + \ell - r + 2}{m} = -\ell + \frac{\ell - r + 2}{m}.$$
Now, $-\ell + \frac{\ell - r + 2}{m} > -\ell -1$ because, from $\ell - r + 2 > \ell - m + 2$, the sequence of inequalities $\frac{\ell - r + 2}{m} > \frac{\ell - m + 2}{m} = \frac{\ell + 2}{m} - 1 > -1$ holds. Moreover, $-2\ell m - m - r + 1 = -m(2\ell + 1) - r + 1 < -\ell -1$, because $\ell>0$ and the facts that $m \geq 2\ell +1$ and $r \geq 0$ imply that $-m(2\ell + 1) - r + 1 \leq -(2\ell + 1)(2\ell + 1) - 0 + 1=-4\ell^2 - 4\ell < -\ell - 1$. Then, the lower bound is
$$ b \geq -\ell + \left\lceil \frac{\ell - r + 2}{m} \right\rceil.$$
Notice that $\frac{\ell - r + 2}{m} \leq 0$ as, under the hypothesis in this case, we have $\ell-r+2 \le \ell - (m - \ell + 1) + 2 = 2\ell - m + 1 \leq 0$. Moreover, $\frac{\ell - r + 2}{m} > -1$ because $\ell - r + 2 > \ell - m + 2 > -m$. Therefore, $\left\lceil \frac{\ell - r + 2}{m} \right\rceil = 0$ and the actual range for $b$ is
$$ -\ell \le b \le r - m - 1,$$
giving $r + \ell - m$ solutions on the vertical line $a= \ell +1$. The rotational invariance by the map $R$ of the solutions proves that
$$\card\left(B_u \cap -qB_u\right) = (2\ell + 1)^2 + 4(r + \ell - m).$$
This concludes the proof.
\end{proof}

Next, we consider a second situation where we are able to compute the cardinality of the set $B_u \cap -qB_u$. As before, we are in the case $s=2\varsigma$. Consider $q$ a prime power (either even or odd) and $n$ an odd integer such that $5 \le n = m^2 + 1$ (for some positive integer $m$) which divides $q^2 + 1$. Assume also that $q \equiv \pm m \pmod n$. Let $\ell$ and $r$ be nonnegative integers with $r < m$ and $0 < \ell m + r \le \frac{n-1}{2}$, and define $u = \ell m + r$.

\begin{pro}\label{prop:hdime}
Let $q, u, \ell, r$, and $m$ be as before. Then, the cardinality of the intersection $B_u \cap -qB_u$ equals
\begin{equation}
\card(B_u \cap -qB_u) =
\begin{cases}
 4\ell^2 & \text{if } r \le \frac{m}{2} - \ell, \\[6pt]
 4\ell^2 + 4\left(r + \ell - \frac{m}{2}\right) & \text{if } \frac{m}{2} - \ell < r \le \frac{m}{2} + \ell + 1, \\[6pt]
 4(\ell + 1)^2 & \text{if } r > \frac{m}{2} + \ell + 1.
\end{cases}
\end{equation}
\end{pro}

\begin{proof}
The proof is similar to the one in \Cref{prop:hdimo}. In this case, since $n$ is odd, $B_u = \left\{ \frac{n-1}{2} - u + 1, \dots, \frac{n-1}{2} + u \right\}$ and it contains exactly $2u$ integers. The set $B_u$ is in bijection with the set of odd elements in $\left[ -(2u-1), 2u-1 \right] \cap \mathbb{Z}$ by means of the map $\mathfrak{b} \mapsto i := 2\mathfrak{b}-n$. Analogously, the elements $x \in -qB_u \cap B_u$ correspond to an integer $j := 2x - n$. It is clear that $x \equiv -q \mathfrak{b} \pmod n$ and, after multiplying by $2$, one deduces that we are looking for pairs of odd integers $(i,j)$ such that $j \equiv -qi \pmod n$ and belong to the square $\left( \left[ -(2u-1), 2u-1 \right] \cap \mathbb{Z} \right)^2$.

By assumption, $q \equiv \pm m \pmod n$, so we wish to find elements in the set
$$ L = \left\{ (i,j) \in \mathbb{Z}^2 \mid j \equiv mi \pmod n \right\}. $$
Since $m^2 \equiv -1 \pmod n$, $L$ is orthogonal and generated by $(1, m)$ and $(-m, 1)$; hence, the above pairs $(i,j)$ can be expressed as $(i, j) = (a - bm, am + b)$ for suitable integers $a,b$. We require $i$ and $j$ to be odd integers, which implies that the integers $a$ and $b$ must be odd because $m$ is even. Recalling that $u = \ell m + r$ for $\ell$ and $m$ as in the statement, we must find the number of pairs of odd integers $(a,b)$ that satisfy
\begin{align}
\left| a - bm \right| &\le 2\ell m + 2r - 1, \label{eq:odd_bound1} \\
\left| am + b \right| &\le 2\ell m + 2r - 1. \label{eq:odd_bound2}
\end{align}

Let us denote by
$$ G \subset \mathbb{Z}^2 $$
the set of such solutions. The points with $\left| a \right|, \left| b \right| \le 2\ell-1$ satisfy the above inequalities by the triangle inequality and the facts that $r \ge 0$ and $\ell < \frac{m}{2}$ (this last inequality holds because $u \le \frac{m^2}{2}$). This allows us to define an initial candidate grid
$$ G_0 = \left\{ (a, b) \in \mathbb{Z}^2 : \left| a \right| \le 2\ell-1, \left| b \right| \le 2\ell-1, \, a \text{ and } b \text{ are odd} \right\} \subseteq G $$
whose cardinality is $(2\ell)^2 = 4\ell^2$. Moreover, the pair of Inequalities \eqref{eq:odd_bound1} and \eqref{eq:odd_bound2} is invariant under the map $R$ defined in the proof of \Cref{prop:hdimo}.

Now, let us show that $G$ does not contain points with $a \ge 2\ell+3$. Inequalities \eqref{eq:odd_bound1} and \eqref{eq:odd_bound2} provide the following bounds for $b$:
\begin{align}
L(a) := \frac{a - (2\ell m + 2r - 1)}{m} &\le b \le \frac{a + 2\ell m + 2r - 1}{m}, \label{eq:boundb1bis} \\
-(2\ell m + 2r - 1) - am &\le b \le 2\ell m + 2r - 1 - am =: U(a). \label{eq:boundb2bis}
\end{align}
The value $L(a)$ increases (and $U(a)$ decreases) as $a$ increases; therefore, it suffices to prove that $L(2\ell+3) > U(2\ell+3)$ to conclude that $G$ contains no points with $a \ge 2\ell+3$. Reasoning by contradiction, if $L(2\ell+3) \le U(2\ell+3)$, then
$$ \frac{2\ell(1-m) + 4 - 2r}{m} \le 2r - 3m - 1, $$
and thus
$$ 3m^2 + m + 4 \le 2r(m+1) + 2\ell(m-1). $$
Since $r \le m-1$ and $\ell < \frac{m}{2}$, we have
$$ 3m^2 + m + 4 < 2(m-1)(m+1) + 2\left(\frac{m}{2}\right)(m-1) = 3m^2 - m - 2, $$
a contradiction to the fact that $m > 0$.

It remains to study the existence of solutions $(a,b) \in G$ with $a = 2\ell+1$. We divide our study into three cases.

\textit{Case 1: $r \le \frac{m}{2} - \ell$.} The upper bound $U(2\ell+1) = 2r - m - 1$ in Inequalities \eqref{eq:boundb2bis} indicates that the maximum allowed value for $b$ is $-2\ell - 1$, which is the only odd integer under consideration satisfying $\left| b \right| \le 2\ell+1$. Then, Condition \eqref{eq:odd_bound1} yields
$$ \left| 2\ell + 1 - (-2\ell - 1)m \right| = 2\ell m + m + 2\ell + 1, $$
which is strictly greater than $2\ell m + 2r - 1$. Therefore,
$$ \card(B_u \cap -qB_u) = \card(G_0) = 4\ell^2. $$

\textit{Case 2: $\frac{m}{2} - \ell < r \le \frac{m}{2} + \ell + 1$.} Once $r > \frac{m}{2} - \ell$, new solutions are captured. First, let us specify the allowable interval for $b$ from Conditions \eqref{eq:boundb1bis} and \eqref{eq:boundb2bis} when $a = 2\ell+1$. Let $U_1 = \frac{2\ell(m+1)+2r}{m}$ and $L_2 = -4\ell m - m - 2r + 1$ denote the upper and lower bounds in Conditions \eqref{eq:boundb1bis} and \eqref{eq:boundb2bis}, respectively. Note that
$$ L(2\ell+1) - L_2 = \frac{2\ell m(2m - 1) + m(m - 1) + 2r(m - 1) + 2\ell + 2}{m}, $$
and the facts that $m \ge 2$ and $\ell, r \ge 0$ ensure that the numerator is strictly positive. Hence, $L(2\ell+1) > L_2$. Moreover,
$$ U_1 - U(2\ell+1) = \frac{(m+2\ell)(m+1) - 2r(m-1)}{m}. $$
Since $2r \le m + 2\ell + 2$ in this case, we deduce that the numerator is greater than or equal to
$$ (m+2\ell)(m+1) - (m + 2\ell + 2)(m-1) = 4\ell+2 > 0. $$
This proves that $U(2\ell+1) < U_1$, and therefore the bounding interval for $b$ given by Conditions \eqref{eq:boundb1bis} and \eqref{eq:boundb2bis} is $\left[ L(2\ell+1), U(2\ell+1) \right]$.

Now, by symmetry across the axes, let us show that the valid odd values for $b$ belong to $\left[ -U(2\ell+1), U(2\ell+1) \right]$. Indeed, $-U(2\ell+1) \ge L(2\ell+1)$ because
$$ -2r + m + 1 \ge \frac{2\ell + 1 - 2\ell m - 2r + 1}{m}, $$
or equivalently,
$$ (m+2)(m-1) + 2\ell(m-1) = m^2 + m - 2\ell - 2 + 2\ell m \ge 2rm - 2r = 2r(m-1), $$
that is, $m + 2 + 2\ell \ge 2r$, which is assumed in this case.

Let us count the new solutions beyond those in Case 1. Let $r = \frac{m}{2} - \ell + r'$ for some $r' > 0$. Substituting into the upper bound, we have $U(2\ell+1) = 2r' - 2\ell - 1$. For every unit increment of $r$ (and thus of $r'$), the upper bound $U(2\ell+1)$ increases by $2$ units, capturing exactly $1$ new valid (odd) integer $b$. The rotational invariance under the map $R$ of the solutions in $G$ shows that each unit increment in $r$ adds $4$ new solutions. Therefore,
$$ \card(B_u \cap -qB_u) = 4\ell^2 + 4r' = 4\ell^2 + 4\left(r - \left(\frac{m}{2} - \ell\right)\right) = 4\ell^2 + 4\left(r + \ell - \frac{m}{2}\right). $$
Finally, let us consider the remaining case.

\textit{Case 3: $r > \frac{m}{2} + \ell + 1$.} Let us show that the bounds for $b$ given by Conditions \eqref{eq:boundb1bis} and \eqref{eq:boundb2bis} contain the interval $\left[ -(2\ell+1), 2\ell+1 \right]$. For the upper bounds, on the one hand, we have $U(2\ell+1) > 2\ell + 1$ since $U(2\ell+1) - (2\ell + 1) = 2r - m - 2\ell - 2 = 2\left(r - \left(\frac{m}{2} + \ell + 1\right)\right) > 0$ by the hypothesis in this case. The other upper bound is
$$ U_1 = \frac{2\ell m + 2\ell + 2r}{m} = 2\ell + \frac{2\ell + 2r}{m}, $$
and the fact that $2r > m + 2\ell + 2$ implies $U_1 > 2\ell + 1 + \frac{4\ell + 2}{m} > 2\ell + 1$.

For the lower bounds, we have $-(2\ell + 1) > L(2\ell+1) = \frac{2\ell + 1 - (2\ell m + 2r - 1)}{m}$, since
$$ -2\ell m - m \ge 2\ell + 1 - 2\ell m - 2r + 1, $$
or equivalently, $2r \ge m + 2\ell + 2$, which is our assumption. Substituting this into the other lower bound $L_2 = -4\ell m - m - 2r + 1$, we obtain
$$ L_2 < -4\ell m - m - (m + 2\ell + 2) + 1 = -4\ell m - 2m - 2\ell - 1 < -(2\ell + 1). $$

Thus, all pairs of odd integers $(a,b)$ with $\left| a \right|, \left| b \right| \le 2\ell + 1$ belong to $G$, giving a total of $4(\ell + 1)^2$ solutions; therefore,
$$ \card(B_u \cap -qB_u) = 4(\ell + 1)^2. $$
\end{proof}

With the help of the two previous propositions, we next provide the analogue of Theorem \ref{thm:qplusone} for the Hermitian inner product. It can be proved in the same way.

\begin{thm}\label{thm:qplusonehermitian}
Keep the notation as at the beginning of Subsection \ref{subsec:BCH}. Assume that $s=2\varsigma$ and $q$ is a prime power. Let $n > 2$ be a positive integer such that $n = m^2 + 1$ for some positive integer $m$, $n$ divides $q^2+1$, and $q \equiv \pm m \pmod n$. Suppose also that if $n$ is even, then $q$ is odd. Let $u$ and $B_u$ be as in \Cref{thm:qminusone}. Then, the code $(\BCH_{q^2}^N(\mathcal{A}+B_u))^{\perp_h}$ is a classical
$$ \left[ N, N-\frac{N}{n}(2u - ((n+1) \bmod 2)), 2u - ((n+1) \bmod 2) + 1 \right]_{q^2} $$
optimal LRC with locality $$(r,\delta) = \left( n-2u + ((n+1) \bmod 2), 2u - ((n+1) \bmod 2) + 1 \right).$$ Therefore, the EAQECC $Q'(\BCH_{q^2}^N(\mathcal{A}+B_u))$ is an optimal pure EA $(r,\delta)$-QLRC with parameters
$$ \left[\left[ N, N-2\frac{N}{n}(2u - ((n+1) \bmod 2))+c, 2u - ((n+1) \bmod 2) + 1; c \right]\right]_q, $$
and locality $(r,\delta) = \left( n-2u + ((n+1) \bmod 2), 2u - ((n+1) \bmod 2) + 1 \right)$. The value $c = \frac{N}{n}\card(B_u \cap -qB_u)$ can be computed, according to the parity of $q$ and $n$, with the formulae in \Cref{prop:hdimo} or \Cref{prop:hdime}.
\end{thm}

\subsection{EA $(r,\delta)$-QLRCs from homothetic-BCH codes}\label{subsec:homo}
Keep the notation as at the beginning of \Cref{sec:BCHhomo}. Consider the sets $\mathcal{A}$, $B_{u,v}$, and $B_{u}$ introduced in Subsection \ref{subsec:BCH} and recall that $\beta, \zeta_n \in \mathbb{F}_{q^s}$ are primitive $N$-th and $n$-th roots of unity, respectively. Denote by $B \subseteq \mathbb{Z}_n$ a $q$-complete set modulo $n$, which will be chosen to be $B_{u,v}$ or $B_{u}$. Let $\ell$ be a positive integer and $\lambda := \frac{N}{n}-\ell$. Consider the homothetic-BCH code
$$ \mathcal{S}_{\mathcal{A}+B}^{P,q} \subseteq \mathbb{F}_q^{N-\ell n}, $$
where
$$ P = \left\{1, \zeta_n, \dots, \zeta_n^{n-1}, \beta, \beta\zeta_n, \dots, \beta\zeta_n^{n-1}, \dots, \beta^{\lambda-1}, \beta^{\lambda-1}\zeta_n, \dots, \beta^{\lambda-1}\zeta_n^{n-1}\right\}. $$
The homothetic-BCH code $\mathcal{S}_{\mathcal{A}+B}^{P,q}$ coincides with the punctured code $\pi_{\left\{1,\dots,N-\ell n\right\}}(\BCH_q^N(\mathcal{A}+B))$ for a certain order of coordinates \cite[Proposition 1.5.8]{HMCthesis}. We will use this order.

\begin{lem}\label{lem:hdim}
The dimension of the homothetic-BCH code $\mathcal{S}_{\mathcal{A}+B}^{P,q}$ equals
\begin{equation}
\dim \left(\mathcal{S}_{\mathcal{A}+B}^{P,q}\right) = \left( \frac{N}{n}-\ell \right) \card(B).
\end{equation}
\end{lem}

\begin{proof}
Consider the subsets $I_j := \left\{ jn + 1, \ldots, jn+n \right\}$, $0 \le j \le \frac{N}{n}-1$, of the set $\left\{1,\dots,N\right\}$. Define $L := I_0 \cup \cdots \cup I_{\frac{N}{n}-\ell-1}$ and $K := I_{\frac{N}{n}-\ell} \cup \cdots \cup I_{\frac{N}{n}-1}$. Then, $\mathcal{S}_{\mathcal{A}+B}^{P,q} = \pi_L(\BCH_q^N(\mathcal{A}+B))$.
By \cite[Proposition 1.5.8]{HMCthesis}, $\pi_{I_j}(\BCH_q^N(\mathcal{A}+B)) \subseteq \mathbb{F}_q^n$ is the BCH code (written using homothetic-BCH code notation)
\[
\mathcal{S}_{\mathcal{A}+B}^{\left\{\beta^j, \beta^j\zeta_n, \beta^j\zeta_n^2, \ldots, \beta^j\zeta_n^{n-1}\right\},q},
\]
which is equal to
\[
\mathcal{S}_{B}^{\left\{\beta^j, \beta^j\zeta_n, \beta^j\zeta_n^2, \ldots, \beta^j\zeta_n^{n-1}\right\},q},
\]
as the set $\mathcal{A}+B$ is congruent to $B$ modulo $n$ by the definition of $\mathcal{A}$.
Then, we have $\dim \pi_{I_j}(\BCH_q^N(\mathcal{A}+B)) = \card(B)$ by the paragraph following \Cref{def:BCH}.

Since
$$ \BCH_q^N(\mathcal{A}+B) \subseteq \pi_L(\BCH_q^N(\mathcal{A}+B)) \times \pi_K(\BCH_q^N(\mathcal{A}+B)), $$
it follows that
\begin{multline*}
\frac{N}{n}\card(B) = \dim \left(\BCH_q^N(\mathcal{A}+B)\right) \le \\ \dim \left(\pi_L(\BCH_q^N(\mathcal{A}+B))\right) + \dim \left(\pi_K(\BCH_q^N(\mathcal{A}+B))\right).
\end{multline*}
Moreover, by the same reasoning, we have
\begin{multline*} \dim \left(\pi_L(\BCH_q^N(\mathcal{A}+B))\right) \le \sum_{j=0}^{\frac{N}{n}-1 -\ell} \dim \left(\pi_{I_j}(\BCH_q^N(\mathcal{A}+B))\right) = \\ \left(\frac{N}{n} -\ell\right)\card(B)
\end{multline*}
and
$$ \dim \left(\pi_K(\BCH_q^N(\mathcal{A}+B))\right) \le \sum_{j=\frac{N}{n} -\ell}^{\frac{N}{n}-1} \dim\left( \pi_{I_j}(\BCH_q^N(\mathcal{A}+B))\right) = \ell\card(B).
$$
Combining the above three expressions concludes the proof.
\end{proof}

We are going to provide EA quantum $(r,\delta)$-LRCs from homothetic-BCH codes of the form $\mathcal{S}_{\mathcal{A}+B}^{P,q}$. Our next result proves that their Euclidean and Hermitian duals are classical $(r,\delta)$-LRCs. This is a direct consequence of the fact that $(\mathcal{S}_{\mathcal{A}+B}^{P,q})^{\perp_e}$ is a shortened code of $\BCH_q^N(\mathcal{A}+B)^{\perp_e}$ by \cite{Pless} (and the same holds for the Hermitian inner product), \Cref{thm:qminusone}, and \Cref{lem:hdim}.

\begin{cor}\label{cor:homoclasrd}
Keep the notation as above. If $B$ is chosen to be $B_{u,v}$ or $B_u$, then the code $(\mathcal{S}_{\mathcal{A}+B}^{P,q})^{\perp_e}$ is a classical
$$ \left[ N-\ell n, \left(\frac{N}{n}-\ell\right)(n-\card(B)), \card(B)+1 \right]_q $$
optimal LRC with locality $(r,\delta) = \left( n-\card(B), \card(B)+1 \right)$.

Similarly, in the case $s=2\varsigma$, the code $(\mathcal{S}_{\mathcal{A}+B}^{P,q^2})^{\perp_h}$ is a classical $q^2$-ary optimal LRC with the same parameters and locality $(r,\delta)$ as above.
\end{cor}

The next proposition shows that the Euclidean and Hermitian duals of the above homothetic-BCH codes are again homothetic-BCH codes, with defining sets as in the BCH case.

\begin{pro}\label{prop10}
Keep the notation as above. Let $B \subseteq \mathbb{Z}_n$ be a $q$-complete set modulo $n$. Then, the following statements hold.
\begin{enumerate}
	\item The Euclidean dual $(\mathcal{S}_{\mathcal{A}+B}^{P,q})^{\perp_e}$ is the code
	$$ (\mathcal{S}_{\mathcal{A}+B}^{P,q})^{\perp_e} = \mathcal{S}_{\mathcal{A}+\left(\{0, \ldots, n-1\} \setminus -B\right)}^{P,q}. $$
	\item In the case $s=2\varsigma$, the Hermitian dual $(\mathcal{S}_{\mathcal{A}+B}^{P,q^2})^{\perp_h}$ is the code
	$$ (\mathcal{S}_{\mathcal{A}+B}^{P,q^2})^{\perp_h} = \mathcal{S}_{\mathcal{A}+\left(\{0, \ldots, n-1\} \setminus -qB\right)}^{P,q^2}. $$
\end{enumerate}
\end{pro}

\begin{proof}
Let $i \in \mathcal{A}+B$ and $j \in \mathcal{A}+ \left(\{0, \dots, n-1\} \setminus -B\right)$. The evaluation vectors $\ev_P(X^i)$ and $\ev_P(X^j)$, where $\ev_P$ denotes the evaluation map in Equation \eqref{homoevamap}, are orthogonal to each other with respect to the Euclidean inner product. This is because, similarly to \cite[Equation (5)]{hbch25}, it holds that
$$ \ev_P(X^i) \cdot_e \ev_P(X^j) = \left( \sum_{k=0}^{n-1}\zeta_n^{k(i+j)} \right) \left( 1+\beta^{i+j}+\cdots+\beta^{(\lambda-1)(i+j)}\right), $$
and the first factor on the right-hand side vanishes.
This implies that $\mathcal{S}_{\mathcal{A}+\left(\{0, \ldots, n-1\} \setminus -B\right)}^{P,q} \subseteq (\mathcal{S}_{\mathcal{A}+B}^{P,q})^{\perp_e}$.
Lemma \ref{lem:hdim} and a dimensionality argument establish equality. The same argument applies to the code $\mathcal{S}_{\mathcal{A}+\left(\{0, \ldots, n-1\} \setminus -B\right)}^{P,q}$ since $\{0, \ldots, n-1\} \setminus -B$ is also $q$-complete modulo $n$.

Statement (2) can be proved analogously, also as a consequence of \cite[Equation (5)]{hbch25}.
\end{proof}

As in \Cref{prop:hulldim} for the BCH construction, one can compute the amount of entanglement $c$, which follows from \Cref{prop10}.

\begin{pro}\label{prop:homohulldim}
Keep the notation as at the beginning of Subsection \ref{subsec:homo}. Let $B \subseteq \mathbb{Z}_n$ be a $q$-complete set modulo $n$. Then, the following statements hold.
\begin{enumerate}
	\item The Euclidean amount of entanglement
	$$ c = \dim \left(\mathcal{S}_{\mathcal{A}+B}^{P,q}\right) - \dim \left(\mathcal{S}_{\mathcal{A}+B}^{P,q} \cap (\mathcal{S}_{\mathcal{A}+B}^{P,q})^{\perp_e}\right) $$
	is equal to $\left(\frac{N}{n}-\ell\right) \card(B \cap -B)$.
	\item In the case $s=2\varsigma$, the Hermitian amount of entanglement
	$$ c = \dim \left(\mathcal{S}_{\mathcal{A}+B}^{P,q^2}\right) - \dim \left(\mathcal{S}_{\mathcal{A}+B}^{P,q^2} \cap (\mathcal{S}_{\mathcal{A}+B}^{P,q^2})^{\perp_h}\right) $$
	equals $\left(\frac{N}{n}-\ell\right) \card(B \cap (-qB))$.
\end{enumerate}
\end{pro}

The purity of the quantum codes provided in the following theorems (arising from $\mathcal{S}_{\mathcal{A}+B}^{P,q}$) is a consequence of the following proposition, which can be proved using arguments similar to those in the proof of Proposition \ref{pro:purity}.

\begin{pro}\label{pro:purityhbch}
Keep the notation as above. Let $B \subseteq \mathbb{Z}_n$ be a nonempty set of consecutive integers that is $q$-complete (respectively, $q^2$-complete in the case $s=2\varsigma$) modulo $n$. Then, we have either
$$ \dis\left(\left( \mathcal{S}_{\mathcal{A}+B}^{P,q} \right)^{\perp_e}\right) = \dis\left( \left( \mathcal{S}_{\mathcal{A}+B}^{P,q} \right)^{\perp_e} \Bigg\backslash \left( \left( \mathcal{S}_{\mathcal{A}+B}^{P,q} \right)^{\perp_e} \cap \mathcal{S}_{\mathcal{A}+B}^{P,q} \right)\right) $$
or
$$ \left( \mathcal{S}_{\mathcal{A}+B}^{P,q} \right)^{\perp_e} = \left( \mathcal{S}_{\mathcal{A}+B}^{P,q} \right)^{\perp_e} \cap \mathcal{S}_{\mathcal{A}+B}^{P,q} $$
(respectively,
$$ \dis\left(\left( \mathcal{S}_{\mathcal{A}+B}^{P,q^2} \right)^{\perp_h}\right) = \dis\left( \left( \mathcal{S}_{\mathcal{A}+B}^{P,q^2} \right)^{\perp_h} \Bigg\backslash \left( \left( \mathcal{S}_{\mathcal{A}+B}^{P,q^2} \right)^{\perp_h} \cap \mathcal{S}_{\mathcal{A}+B}^{P,q^2} \right)\right) $$
or
$$ \left( \mathcal{S}_{\mathcal{A}+B}^{P,q^2} \right)^{\perp_h} = \left( \mathcal{S}_{\mathcal{A}+B}^{P,q^2} \right)^{\perp_h} \cap \mathcal{S}_{\mathcal{A}+B}^{P,q^2} $$
in the case $s=2\varsigma$).
\end{pro}

We conclude this section by stating Theorems \ref{thm:homoqminusone}, \ref{thm:homoqplusone}, and \ref{thm:homoqplusonehermitian}, which are the analogues of Theorems \ref{thm:qminusone}, \ref{thm:qplusone}, and \ref{thm:qplusonehermitian} for the homothetic-BCH case. The proofs follow by using \Cref{cor:homoclasrd}, Propositions \ref{prop10}, \ref{prop:homohulldim}, and \ref{pro:purityhbch}, and arguments similar to those used in the proofs for the BCH case.

\begin{thm}\label{thm:homoqminusone}
Assume that $n \ge 2$ divides $q-1$. Then, the code $(\mathcal{S}_{\mathcal{A}+B_{u,v}}^{P,q})^{\perp_e}$ is a classical
$$ \left[ N-\ell n, \left(\frac{N}{n}-\ell\right)(n-(v-u+1)), v-u+2 \right]_q $$
optimal LRC with locality $(r,\delta) = (n-v+u-1, v-u+2)$.

Therefore, the EAQECC $Q'(\mathcal{S}_{\mathcal{A}+B_{u,v}}^{P,q})$ is an optimal pure EA $(r,\delta)$-QLRC with parameters
$$ \left[\left[ N-\ell n, N-\ell n - 2\left(\frac{N}{n}-\ell\right)(v-u+1) + c, v-u+2; c \right]\right]_q, $$
where $c = \left(\frac{N}{n}-\ell\right)\max\left\{0, 1 + \min\{n-2u, 2v - n\}\right\}$, and locality $(r,\delta) = (n-v+u-1, v-u+2)$.
\end{thm}

\begin{thm}\label{thm:homoqplusone}
Assume that $n \ge 2$ divides $q+1$. Then, the code $(\mathcal{S}_{\mathcal{A}+B_u}^{P,q})^{\perp_e}$ is a classical
$$ \left[ N-\ell n, \left(\frac{N}{n}-\ell\right)(n-(2u -((n+1) \bmod 2))), 2u - ((n+1) \bmod 2) + 1 \right]_q $$
optimal LRC with locality $$(r,\delta) = (n-2u + ((n+1) \bmod 2), 2u - ((n+1) \bmod 2) + 1).$$

Therefore, the EAQECC $Q'(\mathcal{S}_{\mathcal{A}+B_u}^{P,q})$ is an optimal pure EA $(r,\delta)$-QLRC with parameters
\begin{multline*}
  \Bigl[\Bigl[ N-\ell n, N-\ell n - \left(\frac{N}{n}-\ell\right)(2u - ((n+1) \bmod 2)), \\
  2u - ((n+1) \bmod 2) + 1; c \Bigr]\Bigr]_q,
\end{multline*}
where $c = \left(\frac{N}{n}-\ell\right)(2u - ((n+1) \bmod 2))$, and locality $$(r,\delta) = (n-2u + ((n+1) \bmod 2), 2u - ((n+1) \bmod 2) + 1).$$
\end{thm}

\begin{thm}\label{thm:homoqplusonehermitian}
Keep the notation as above. Assume that $s=2\varsigma$ and $2 < n = m^2+1$ for some positive integer $m$, where $n$ divides $q^2+1$ and $q \equiv \pm m \pmod n$. Suppose also that if $n$ is even, then $q$ is odd. Then, the code $(\mathcal{S}_{\mathcal{A}+B_u}^{P,q^2})^{\perp_h}$ is a classical
$$ \left[ N-\ell n, \left(\frac{N}{n}-\ell\right)(n-(2u - ((n+1) \bmod 2))), 2u - ((n+1) \bmod 2) + 1 \right]_{q^2} $$
optimal LRC with locality $$(r,\delta) = (n-2u + ((n+1) \bmod 2), 2u - ((n+1) \bmod 2) + 1).$$

Therefore, the EAQECC $Q'(\mathcal{S}_{\mathcal{A}+B_u}^{P,q^2})$ is an optimal pure EA $(r,\delta)$-QLRC with parameters
\begin{multline*}
  \Bigl[\Bigl[ N-\ell n, N-\ell n - 2\left(\frac{N}{n}-\ell\right)(2u - ((n+1) \bmod 2)) + c, \\
  2u - ((n+1) \bmod 2) + 1; c \Bigr]\Bigr]_q,
\end{multline*}
and locality $(r,\delta) = (n-2u + ((n+1) \bmod 2), 2u - ((n+1) \bmod 2) + 1)$. The amount of entanglement $c = \left(\frac{N}{n}-\ell\right)\card(B_u \cap -qB_u)$ can be computed, according to the parity of $q$ and $n$, with the formulae given in \Cref{prop:hdimo} or \Cref{prop:hdime}.
\end{thm}

\begin{rem}
Note that when $N-\ell n$ and $q$ are relatively prime, one can instead use the BCH constructions in Subsection \ref{subsec:BCH} to obtain the same parameters as the codes provided in this section.
\end{rem}

\section*{Data availability}
No databases were generated or analyzed during this study.

\section*{Conflict of interest}
The authors declare they have no conflict of interest.

\section*{Use of AI}
Some ideas behind the proofs of Theorem \ref{thm:1} and Propositions \ref{prop:hdimo} and \ref{prop:hdime} were obtained through interaction with AI (Google Gemini 3.1). No AI-generated texts remain in the paper.

%\clearpage
%\bibliographystyle{plain}
%\bibliography{biblio-ryutaroh1}

\hfill\break

\end{document}